\documentclass{article}
\usepackage[utf8]{inputenc}

\usepackage[style=phys, articletitle=true, biblabel=brackets, backend=biber]{biblatex}
\usepackage{amsfonts}
\usepackage{amsmath}
\usepackage{amssymb}
\usepackage{amsthm}
\usepackage{leftindex}
\usepackage{dsfont}
\usepackage{cancel}
\usepackage{physics}
\usepackage{microtype}
\usepackage{hyperref}
\usepackage{algorithm}
\usepackage{tikz}
\usepackage{amsmath} 
\usepackage{amsfonts} 
\usepackage{color}
\usepackage{xcolor}
\usepackage{bm}
\usepackage{booktabs}
\usepackage{mathtools}
\usepackage{enumitem}
\usetikzlibrary{shapes, arrows.meta, positioning, calc, decorations.pathmorphing}

\tikzset{
  quantum_box/.style={
    rectangle, rounded corners=2mm, draw=blue, line width=0.8pt,
    fill=blue!20, inner sep=2mm, minimum width=1.8cm, minimum height=1cm, align=center
  },
  classical_box/.style={
    rectangle, rounded corners=2mm, draw=blue, line width=0.8pt, dotted,
    fill=blue!20, inner sep=2mm, minimum width=1.8cm, minimum height=1cm, align=center
  },
  quantum_state/.style={
    circle, draw=red, line width=0.8pt, fill=red!20,
    inner sep=1mm, minimum width=0.5cm, minimum height=0.5cm, align=center
  },
  classical_line/.style={
    draw=black, line width=0.8pt,   dotted, -Latex, shorten >=1mm, shorten <=1mm
  },
  classical_line_dotted/.style={
    classical_line, dotted
  },
  quantum_line/.style={
    draw=red, line width=0.8pt, -Latex, shorten >=1mm, shorten <=1mm
  },
  label_style/.style={font=\small}
}

\newtheorem{example}{Example}  
\usepackage{subcaption}
\usepackage[a4paper, total={6in, 9.5in}]{geometry}
\usepackage{authblk}

\newtheorem{theorem}{Theorem}
\newtheorem*{theoremcopy}{Theorem}
\newtheorem{Proposition}{Proposition}
\newtheorem{Definition}{Definition}

\title{Structure of quantum measurements implementable with one round of classical communication}
\author[1,2]{Arthur C. R. Dutra} 
\author[1,3,4]{Ties-A. Ohst\thanks{\texttt{Corresponding author, e-mail: ties.ohst@physics.uu.se}}}
\author[1]{Hai-Chau Nguyen}
\author[1]{Otfried Gühne}
\affil[1]{Naturwissenschaftlich-Technische Fakultät, 
Universität Siegen, Walter-Flex-Straße 3, 57068 Siegen, Germany}
\affil[2]{Instituto de Matemática, Estatística e Computação Científica,
Universidade Estadual de Campinas, Campinas, SP, Brazil}
\affil[3]{Department of Physics and Astronomy, Uppsala University, 75120 Uppsala, Sweden}
\affil[4]{Nordita, KTH Royal Institute of Technology and Stockholm University, 10691 Stockholm, Sweden}

\begin{document}
\maketitle
\begin{abstract}
        Measurements that can be implemented via local operations and classical communication (LOCC) constitute a class of operations that is available in future quantum networks in which parties 
        share entangled resource states. 
        We characterise the different classes of measurements implementable with LOCC, where communication is restricted to 
        a single round with a fixed direction.  
        In particular, using the framework of constrained separability problems, we provide a complete characterisation of the class of LOCC measurements that require one round of classical communication with a limit on the transmitted information. Furthermore, we show how to distinguish between adaptive and non-adaptive measurement strategies. Using our techniques, we present examples where the success probability of state discrimination depends on the direction of communication as well as on the message size. We also discuss explicit instances of state ensembles where 
        non-projective measurements provide an advantage and where adaptive measurement strategies lead to improved success rates when compared to all non-adaptive strategies.
\end{abstract}

\section{Introduction}

A wide range of quantum technologies involve two observers implementing measurements at distant locations~\cite{ekkert91, teleportation, remote_preparation, data_hiding}. Although arbitrary joint measurements across the parts can often be powerful, practical technological limitations typically restrict the parties to particular subclasses of joint measurements. 
The most famous of such measurement classes is implemented by Local Operations and Classical Communication (LOCC)~\cite{locc_review} in which the parties can perform arbitrary measurements in their local subsystem that may depend on previous measurement outcomes of the other parties. We will refer to them as LOCC measurements. 

This motivated efforts to characterise the set of LOCC measurements and its achievable advantages. However, LOCC measurements are notoriously difficult to handle computationally: the set is non-convex and the constraints are non-linear, complicating numerical approaches. As a result, most works either replace LOCC with a tractable convex relaxation such as positive partial transpose or separable measurements~\cite{ppt_relax_1,ppt_relax_2, discrimination_thiesis} or seek analytical solutions for specific problems~\cite{nonlocality_without_entanglement, hardy_bipartite,2x2_distinguishability,3x3_distinguishability,2x3_distinguishability}.

Beyond deciding whether a task is implementable by LOCC, we may further ask what classical resources are required to implement it within an LOCC protocol. This classical budget comprises how many communication rounds are needed, how many bits are sent in each round, and the measurement order. Questions of this type have also been recently addressed in the context of quantum state verification \cite{Yu2022}. From an engineering perspective, the use of classical resources can translate into concrete costs, such as higher latency and longer memory hold times. While there are results on the required number of rounds \cite{two_way_required, infinite_rounds,bound_discrimination, wild_unlimited_com}, a comprehensive account of the overall classical resources remains open.

Recently, a new technique for optimisation over bilinear forms \cite{berta_optimization} has been applied to similar problems in the detection of quantum memory and its role in quantum steering \cite{steering_certification, spontaneous_emissions, channel_discrimination,wild_unlimited_com}. In this work, we build upon this approach to construct a converging semidefinite program (SDP) hierarchy that can be used for an optimisation over one-round LOCC measurements with a fixed communication budget. These methods go beyond the standard approach of only demanding that all measurement operators are separable, which has often been used to describe all LOCC measurements \cite{locc_review, Duan2009, Matthews2009, Bandyopadhyay2015}; since our method can constrain the amount and direction of communication it even goes beyond the recent method for the characterisation of the whole set of one-round LOCC measurements introduced in \cite{wild_unlimited_com}.
Concretely, with our hierarchies we can expose the operational parameters of LOCC as explicit constraints, which allow us to tackle the following questions:
\begin{enumerate}[label=(\alph*)]
    \item Are multiple rounds of classical communication required?
    \item How many bits must be exchanged in the communication?
    \item Does the order of local measurements affect the performance? If so, who should measure first?
    \item Is adaptivity necessary for the second observer, that is, must their measurement depend on the first observer's outcome?
\end{enumerate}

As a case study, this paper focuses on the application of LOCC measurements to the problem of locally distinguishing quantum states~\cite{discrimination_locc_review}. 

This paper is structured as follows. In Section \ref{sec:state_dis} we give a brief introduction to the task of minimum-error state discrimination before we discuss LOCC measurements in Section \ref{sec:locc}. Specifically, we will provide complete characterisations of the set of general bipartite one-round LOCC measurements with fixed communication bounds in Section \ref{sec:1r_locc} and of the set of non-adaptive LOCC measurements in Section \ref{sec:na_meas}. In Section \ref{sec:examples}, our methods are applied to three different examples of state discrimination problems, respectively giving rise to particularly interesting answers to the questions (a)--(d) raised above. Our work is concluded in Section \ref{sec:conclusion} and explicit proofs of our two main theorems are provided in the appendices.

\section{Quantum State Discrimination}
\label{sec:state_dis}

The operational task we consider is \textit{minimum-error state discrimination}~\cite{Helstrom1976-wb, discrimination_review, discrimination_thiesis, bound_discrimination}. Here, a state $\varrho_{\lambda}$ described by a density matrix acting on a finite-dimensional Hilbert space is drawn from a finite set $\{\varrho^\lambda\}_{\lambda =1}^{n}$ of states together with associated prior probabilities $\{p_\lambda\}_{\lambda =1}^{n}$. The problem now lies in the construction of the optimal discrimination strategy, described by a measurement that is mathematically given by a POVM, i.e., a collection of positive semidefinite operators $\{M^{\lambda}\}_{\lambda = 1}^{n}$ such that $M^{\lambda} \succcurlyeq 0$ and $\sum_{\lambda=1}^{n} M^\lambda = \mathds{1}$ to determine the unknown parameter $\lambda$.

Given that the measurement is subject to some restriction described as $\{M^\lambda\}_{\lambda} \in \mathcal{M}$, where $\mathcal{M}$ some closed subset of all possible measurements, the maximal success probability in the state discrimination problem can be formulated as the optimisation problem
\begin{equation}
    \label{eq:state_discrimination}
    p^\text{succ}_{\mathcal{M}} = \underset{\{M^{\lambda}\} \in \mathcal{M}} {\textnormal{max}} \sum_{\lambda = 1}^{n} p _\lambda \Tr(M^{\lambda} \varrho_{\lambda})
\end{equation}
and the corresponding minimal error is given by $p^{\text{err}}_{\mathcal{M}} = 1- p^\text{succ}_{\mathcal{M}}$. It is well known, see for instance Ref.~\cite{discrimination_locc_review}, that in the unrestricted case, which $\mathcal{M}$ contains all possible POVMs, the value $p^\text{succ}_{\mathcal{M}}$ can be computed efficiently by a semidefinite program. For general subsets $\mathcal{M}$ however, the problem of determining $p^\text{succ}_{\mathcal{M}}$ may be difficult. 

Since the objective function of the minimum error state discrimination problem in Eq.~\eqref{eq:state_discrimination} is linear, it should be noted that one can always replace the closed and hence also compact set $\mathcal{M}$ of admissible measurements by its convex hull $\textnormal{conv}(\mathcal{M})$, i.e., 
\begin{equation}
    \underset{\{M^{\lambda}\} \in \mathcal{M}} {\textnormal{max}} \sum_{\lambda = 1}^{n} p _\lambda \Tr(M^{\lambda} \varrho_{\lambda}) = \underset{\{M^{\lambda}\} \in \textnormal{conv}(\mathcal{M})} {\textnormal{max}} \sum_{\lambda = 1}^{n} p _\lambda \Tr(M^{\lambda} \varrho_{\lambda})
\end{equation}
as the problem on the right-hand side is guaranteed to be solved optimally by an extreme point of $\textnormal{conv}(\mathcal{M})$ that is automatically an element of $\mathcal{M}$. In this way, the closedness of $\mathcal{M}$ and the linearity of the objective function allow us to formulate the task of finding the optimal success probability of discrimination $p^\text{succ}_{\mathcal{M}}$ as a convex optimisation problem.

\section{LOCC measurements}
\label{sec:locc}

\begin{figure}
    \centering
    \includegraphics[width=0.7\linewidth]{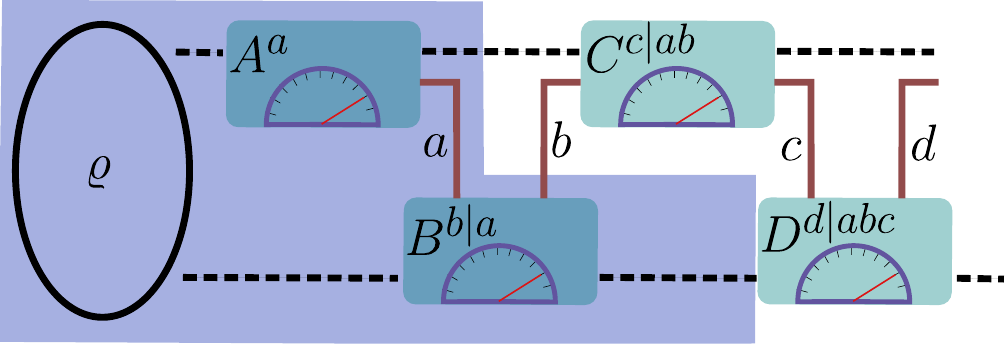}
    \caption{Schematic sketch of general bipartite LOCC measurements. The parties Alice and Bob perform local measurements described by instruments whose specific settings may depend on all preceding measurement results that are freely communicated between Alice and Bob. In this work, we consider the particular case of measurements that only require a single round of classical communication with messages of bounded size, as depicted in the shaded region of the figure.}
    \label{fig:locc_sketch}
\end{figure}

Following the presentation in Ref.~\cite{locc_review}, a bipartite measurement $\{M^{\lambda}\}_{\lambda=1}^{n}\subset \textnormal{Herm}(\mathbb{C}^{d_{\rm A}} \otimes \mathbb{C}^{d_{\rm B}})$ is a LOCC measurement if it can be realized by a sequence of local measurements mathematically described by so-called quantum instruments that may depend on all preceding measurement outcomes that are freely communicated between the parties; see Figure \ref{fig:locc_sketch}. Importantly, while the set of all LOCC measurements does not meet the requirement of being closed \cite{locc_review}, we consider here certain subsets of LOCC measurements that are. A well-known property that all $\rm{LOCC}$ measurements -- independent of the amount and number of rounds of classical communication -- share is the separability of the measurement operators, meaning that 
\begin{equation}
    M^\lambda = \sum_{r} q_{r\lambda} \, A^{r\lambda} \otimes B^{r\lambda},
\end{equation}
where $\{q_{r,\lambda}\}_{r}$ is a probability distribution for each $\lambda$ and $A^{r\lambda} \in \textnormal{Herm}(\mathbb{C}^{d_{\rm A}})$ and $B^{r\lambda} \in \textnormal{Herm}(\mathbb{C}^{d_{\rm B}})$ are any positive semidefinite operators. This
defines the set of separable measurements $\rm{SEP}$ that can be further relaxed to the set $\rm{PPT}$ consisting of measurements in which every measurement operator has positive semidefinite partial transpose \cite{Horodecki1996, Ghne2009}. The sets $\rm{SEP}$ and $\rm{PPT}$ are proper outer approximations of $\rm{LOCC}$ in the sense that there exist separable measurements without an LOCC implementation \cite{locc_review}. 

While the relaxations $\rm{SEP}$ and $\rm{PPT}$ give rise to useful necessary conditions for the characterisation of $\rm{LOCC}$ measurements, a remaining problem lies in the internal distinction between different classes of $\rm{LOCC}$ measurements in terms of their amounts and directions of communication and number of rounds. To overcome this issue, here we contribute a complete characterisation of $\rm{LOCC}$ measurements with a single round of classical communication and bounds on the transmitted communication. This will allow us to systematically study the role of classical communication in such measurements. 

\subsection{One-round LOCC measurements with communication bounds}
\label{sec:1r_locc}
A special class of LOCC measurements is given by one-round LOCC measurement in which Alice performs a local measurement $\{A^{a}\}_{a=1}^{m}$ whose outcome $a$ is subsequently communicated to Bob. Since the parties individually only perform a single measurement, it suffices to describe the local measurements by POVMs instead of using non-destructive instruments that would be needed for more than one round of classical communication. Dependent on the message $a$, Bob performs a measurement $\{B^{\lambda|a}\}_{\lambda=1}^{n}$ and the resulting outcome determines which state of the ensemble is guessed. Formally, this one-round strategy is defined as follows.  
\begin{Definition}
    A bipartite $\textnormal{POVM}$ $\{M^{\lambda}\}_{\lambda = 1}^{n} \subset \textnormal{Herm}(\mathbb{C}^{d_{\rm A}} \otimes \mathbb{C}^{d_{\rm B}})$ is called a one-round $\textnormal{LOCC}$ measurement if it can be decomposed as 
    \begin{equation}
        \label{eq:1-r-meas-op}
        M^{\lambda} = \sum_{a=1}^{m}  A^{a} \otimes B^{\lambda|a},
    \end{equation}
    where $\{A^{a}\}_{a=1}^{m}$ and $\{B^{\lambda|a}\}_{\lambda=1}^{n}$ are local $\textnormal{POVM}$'s in Alice's and Bob's systems, respectively, and $m \in \mathbb{N}$.  We define \(\mathrm{1R\text{-}LOCC}\) as the set of all such POVMs for any number of outcomes $n$ satisfying Eq.~\eqref{eq:1-r-meas-op}. For a fixed size $m$ of the message that Alice sends to Bob, we denote the set of one-round $\textnormal{LOCC}$ measurements with bounded communication as $1\textnormal{R-LOCC}_{m}$.
\end{Definition}

The sets of LOCC measurements form a nested family of subsets in terms of the number of rounds and amount of communication \cite{locc_review}
\begin{equation}
    \text{1R-LOCC}_m\subseteq\text{1R-LOCC} \subsetneq \text{LOCC} \subsetneq \text{SEP} \subseteq \text{PPT},
\end{equation}
for $d_{\rm A},d_{\rm B} \geq 2$, where the leftmost subset inclusion 
is reduced to equality for $m\geq n(d_{\rm A}d_{\rm B})^8$~\cite{locc_review}, while the rightmost inclusion between $\text{SEP}$ and $\text{PPT}$ reduces to equality only if $d_{\rm A}d_{\rm B} \leq 6$ \cite{Horodecki1996} as it has been derived by Ryszard Horodecki
and his sons Michał and Paweł. 

One important remark is that, for any $m\in \mathbb{N}$, the set $\text{1R-LOCC}_m$ is compact \cite{locc_review}. This means we can relax the set by its convex hull when considering its role in minimum error state discrimination.  
In the following, we are going to discuss how the convex hull of the set of one-round $\textnormal{LOCC}$ measurements with bounded communication may be characterised by a converging hierarchy of outer approximations that each may be characterised in terms of positive semidefinite matrices subject to affine-linear constraints. 
Recall that, as we consider the problem of quantum state discrimination as discussed in Section \ref{sec:state_dis}, the relaxation to the convex hull $\text{conv}(\text{1R-LOCC}_m)$ gives rise to the same success probability.   
The idea is that one can apply the recently developed hierarchy of constrained symmetric extensions \cite{berta_optimization} for bipartite separable quantum states whose tensor factors are subject to additional affine constraints. 

The idea of our method starts with the introduction of operators $R^{a\lambda b} \in \textnormal{Herm}(\mathbb{C}^{d_{\rm A}} \otimes \mathbb{C}^{d_{\rm B}})$ defined by 
\begin{equation}
    \label{eq:simple_r_representation}
    R^{a\lambda b} = A^{a} \otimes B^{\lambda|b}
\end{equation}
that allow us to rewrite Eq.~\eqref{eq:1-r-meas-op} via
\begin{equation}
    M^{\lambda} = \sum_{a=1}^{m} R^{a\lambda a}. 
\end{equation}
From Eq.~\eqref{eq:simple_r_representation}, the operator $R^{a\lambda b}$ satisfies some constraints. First, it is clearly a positive semidefinite and separable operator, which implies that it has a positive semidefinite partial transpose, i.e., $(R^{a \lambda b})^{T_{\rm A}} \succcurlyeq 0$. Furthermore, summing over the indices $\lambda$ and $a$ gives rise to the equalities
\begin{align}
    &\sum_{a=1}^{m} R^{a \lambda b}=\frac{\mathds{1}_{\rm A}}{d_{\rm A}} \otimes \sum_{a=1}^{m}\Tr_{\rm A}(R^{a \lambda b})\quad \forall \lambda,b, \label{eq:simple_affine_a}\\
& \sum_{\lambda=1}^n R^{a\lambda b}= \sum_{\lambda=1}^{n} \sum_{b' = 1}^{m} \Tr_{\rm B}(R^{a \lambda b'})\otimes \frac{\mathds{1}_{\rm B}}{md_{\rm B}}\quad \forall a,b \label{eq:simple_affine_b}
\end{align}
that follow from the fact that $A^{a}$ and $B^{\lambda|b}$ are measurement operators so that they sum up to the identity. 
With that, we already arrive at the first level $(k=1)$ of the hierarchy used to characterise $\textnormal{conv}(\text{1R-LOCC}_m)$.  Explicitly, in the problem of minimum error state discrimination, one obtains a relaxation by the following optimisation problem.

\begin{algorithm}
    \label{alg:1r_locc}
   $$\begin{array}{ll}
\text{given} & \{(p_\lambda, \varrho_\lambda)\}_{\lambda=1}^n, m \\    
\text{maximize} & \sum_{\lambda=1}^n p_\lambda \text{Tr}(\varrho_\lambda M^\lambda) \\
\text{subject to} &  M^\lambda \coloneqq  \sum_{a=1}^mR^{a \lambda a}\\
&\sum_{a=1}^{m} R^{a \lambda b}=\mathds{1}_{\rm A}/d_{\rm A} \otimes \sum_{a=1}^{m}\Tr_{\rm A}(R^{a \lambda b})\quad \forall \lambda,b\\
& \sum_{\lambda=1}^n R^{a\lambda b}= \sum_{\lambda=1}^{n} \sum_{b' = 1}^{m} \Tr_{\rm B}(R^{a \lambda b'})\otimes \mathds{1}_{\rm B}/md_{\rm B}\quad \forall a,b\\
& R^{a\lambda b}   \succcurlyeq  0\quad \forall a, \lambda,b\\
& (R^{a \lambda b})^{T_{\rm A}} \succcurlyeq 0 \quad\forall a, \lambda,b
\end{array}$$
\caption{$\text{1R-LOCC}_m$ SDP for minimum error discrimination at $k=1$}
\label{algo: LOCC}
\end{algorithm}

The value computed by the optimisation problem \ref{algo: LOCC} gives rise to an efficiently computable upper bound of $p^\text{succ}_{\mathcal{M}}$ with $\mathcal{M} = \text{1R-LOCC}_m$. However, as we will see later in the case of the examples presented in Section \ref{sec:examples}, the quality of this bound is often not sufficient to get the precise value of $p^\text{succ}_{\mathcal{M}}$ or to see interesting separations between different measurement classes.   

To improve the relaxation above, we go one step further and introduce a hierarchy involving the operators $R^{\bm{a} \lambda b} \in \textnormal{Herm}(\mathbb{C}^{d_{\rm A}^{k}} \otimes \mathbb{C}^{d_{\rm B}})$ with $\bm{a} = (a_1,\dots, a_k) \in \{1,\dots,m\}^{\times k}$ that are given by
\begin{equation}
    \label{eq:big_r_rep}
    R^{\bm{a} \lambda b} = A^{a_1} \otimes \dots \otimes A^{a_k} \otimes B^{\lambda|b}
\end{equation}
and $k \in \mathbb{N}$ is called the level of the hierarchy. Since the measurements $\{A^{a_{j}}\}_{a_{j} = 1}^{m}$ in the $k$ tensor factors of $\mathbb{C}^{d_{\rm A}}$ in Eq.~\eqref{eq:big_r_rep} are identical, one can see that the operator $R^{\bm{a} \lambda b}$ is invariant under permutations in the sense that
\begin{equation}
    (U_{k,d_{\rm A}}^{\sigma} \otimes \mathds{1}_{\rm B}) \,  R^{\bm{a}\lambda b} \, (U_{k,d_{\rm A}}^{{\sigma}^{-1}} \otimes \mathds{1}_{\rm B}) = R^{\sigma (\bm{a}) \lambda b}
\end{equation}
where $U_{k,d}^{\sigma}$ is the standard unitary representation of the symmetric group $S_k$ on ${(\mathbb{C}^{d})}^{\otimes k}$ defined by
\begin{equation}
    U_{k,d}^{\sigma} \ket{i_1} \otimes \dots \otimes \ket{i_k} = \ket{i_{\sigma^{-1}(1)}} \otimes \dots \otimes \ket{i_{\sigma^{-1}(k)}}
\end{equation}
for each permutation $\sigma \in S_k$. Furthermore, $R^{\bm{a}\lambda b}$ obeys some affine constraints obtained from summing over the indices $a_1$ and $\lambda$ that generalise the equations \eqref{eq:simple_affine_a} and \eqref{eq:simple_affine_b} in the case of $k=1$. 
With that, we can formulate one of our main contributions in this work given by the following characterisation of $\text{conv}(\text{1R-LOCC}_m)$, i.e, of the convex hull of one-round LOCC measurements with communication bounds. 
\begin{theorem}[$1\textnormal{R}-\textnormal{LOCC}_{m}$ hierarchy]
    \label{thm:1W_locc_m_fixed_char}
    A bipartite measurement $\{M^{\lambda}\}_{\lambda = 1}^{n} \subset \textnormal{Herm}(\mathbb{C}^{d_{\rm A}} \otimes \mathbb{C}^{d_{\rm B}})$ lies in the convex hull of $1\textnormal{R}-\textnormal{LOCC}_{m}$ if and only if for each level $k \in \mathbb{N}$ there exists an array of positive semidefinite operators $R^{\bm{a}\lambda b} \in \textnormal{Herm}(\mathbb{C}^{d_{\rm A}^{k}} \otimes \mathbb{C}^{d_{\rm B}})$, for $\bm{a} = (a_1,\dots, a_k) \in \{1,\dots,m\}^{\times k}$, $\lambda \in \{1,\dots,n\}$ and $b \in \{1,\dots,m\}$, that satisfies the following constraints:
    \begin{align}
        & M^{\lambda} = \frac{1}{d_{\rm A}^{k-1}}\sum_{a_{1}\dots,a_{k} = 1}^{m} \Tr_{\rm{A}_{1}\dots\rm{A}_{k-1}}(R^{a_{1}\dots a_{k}\lambda a_{k}}) \textit{ for all } \lambda \label{eq:1W_cond_start}\\
        & \sum_{a_1=1}^{m} R^{a_1\dots a_k \lambda b} = \frac{\mathds{1}_{\rm A}}{d_{\rm A}} \otimes \sum_{a_1 = 1}^{m}\Tr_{\rm A_{1}}( R^{a_1\dots a_k \lambda b}) \textit{ for all } a_2,\dots, a_k, \lambda, b \label{eq:left_array_cons}\\
        & \sum_{\lambda = 1}^{n} R^{\bm{a}\lambda b} = \sum_{\lambda=1}^{n} \sum_{b'=1}^{m} \Tr_{\rm B}(R^{\bm{a}\lambda b'}) \otimes \frac{\mathds{1}_{\rm B}}{md_{\rm B}} \textit{ for all } \bm{a},b \label{eq:right_array_cons}\\
        & (U_{k,d_{\rm A}}^{\sigma} \otimes \mathds{1}_{\rm B}) \,  R^{\bm{a}\lambda b} \, (U_{k,d_{\rm A}}^{{\sigma}^{-1}} \otimes \mathds{1}_{\rm B}) = R^{\sigma (\bm{a}) \lambda b} \textit{ for all } \sigma, \bm{a}, \lambda, b \label{eq:array_symmetry}\\
        & {(R^{\bm{a}\lambda b})}^{T_{\rm A^{\ell}}} \succcurlyeq 0 \textit{ for all } \bm{a}, \lambda, b \textit{ and } \ell \in \{0,\dots,k\} \label{eq:array_ppt}\\
        & \sum_{a_1 \dots a_k=1}^{m} \sum_{\lambda = 1}^{n}\Tr( R^{\bm{a}\lambda b}) = d_{\rm A}^k d_{\rm B} \textit{ for all } b. \label{eq:1W_cond_end}
    \end{align}
    In Eq.~\eqref{eq:array_ppt}, $X^{T_{\rm A^{\ell}}}$ denotes the partial transpose applied on the first $\ell$ copies of the system $\rm{A}$ with Hilbert space $\mathbb{C}^{d_{\rm A}}$ and in Eq.~\eqref{eq:array_symmetry}, $\sigma(\bm{a})$ is shorthand notation for the tuple $(a_{\sigma^{-1}(1)},\dots,a_{\sigma^{-1}(k)})$. 
   \end{theorem}
   \begin{proof}
       The proof is presented in Appendix~\ref{ap:theorem_1}.
   \end{proof}
   
    Let us now take a moment to clarify why this theorem is useful. For each $k \in \mathbb{N}$, measurements $\{M^{\lambda}\}_{\lambda = 1}^{n} \subset \textnormal{Herm}(\mathbb{C}^{d_{\rm A}} \otimes \mathbb{C}^{d_{\rm B}})$ that satisfy Eq.~\eqref{eq:1W_cond_start}-\eqref{eq:1W_cond_end} form an outer approximation of $\textnormal{conv}(1\textnormal{R}-\textnormal{LOCC}_{m})$, i.e., the convex hull of the one-round LOCC measurements with a communicated message whose size is quantified by $m$. These approximations become monotonically tighter by increasing the parameter $k$ and can be evaluated by a semidefinite program in at most $\mathcal{O}(m^{k+1}nk)$ positive semidefinite matrices, each having $d_{\rm A}^k d_{\rm B}$ rows and columns. In practice, the number of variables can be drastically reduced in actual implementation by carefully exploiting the permutation invariance described by the constraint \eqref{eq:array_symmetry}. Furthermore, these approximations converge to $\textnormal{conv}(1\textnormal{R}-\textnormal{LOCC}_{m})$ in the limit $k \rightarrow \infty$ which is precisely the statement of Theorem \ref{thm:1W_locc_m_fixed_char}. Such a converging sequence is usually referred to as a hierarchy of outer approximations \cite{Doherty2004}.
    While the number of variables strongly rises by increasing $k$, we will later on see that for our examples, interesting conclusions may be drawn already at fairly low levels $k$ of the hierarchy. 
    
    It should also be noticed that an alternative converging symmetric extension hierarchy could also be constructed by taking copies of the measurement operators on Bob's side. The main reason why we applied the extension on Alice's side, as the party that measures first, is a smaller growth in variables when increasing the level $k$ of the hierarchy. This saving is also justified by the existing result on the convergence speed of constrained symmetric hierarchies \cite[Theorem 2.3.]{berta_optimization} which suggests that extending the system of higher dimension does in general not provide a substantial advantage in the convergence speed.  

 Since the objective function for the minimum error discrimination problem is bi-linear in the measurements of Alice and Bob, an optimum is attained with an extremal POVM on Alice's side. Furthermore, it is a well-known fact that an extremal POVM on $\mathbb{C}^{d}$ can have at most $d^2$ non-zero elements \cite{DAriano2011}. Hence, the optimal discrimination probability saturates for $m \ge d_A^2$, and running the hierarchy for $1 < m \le d_A^2$ is sufficient for drawing conclusions for the entire 1R-LOCC set.

\subsection{Non-adaptive LOCC measurements}
\label{sec:na_meas}

Above, we considered measurements that may be implemented by local operations and one round of classical communication. An intriguing aspect in the study of such measurements is the possibility that one party, Bob, may adapt his measurement apparatus based on the information that has been communicated by Alice. A practical problem that arises in such a protocol is that the local measurements must necessarily be performed after each other so that the quantum system must be coherently stored for this communication time. Hence, one requires the presence of a quantum memory, which is still an active field of research due to its challenging technical implementation \cite{Heshami2016}. This problem motivates the question about the distinction between measurements in which Bob is required to perform an adaptation or not \cite{adapted_advantage}. We will refer to the latter as the set of non-adaptive \text{LOCC} measurements defined as follows.  

\begin{Definition}
A bipartite measurement $\{M^{\lambda}\}_{\lambda = 1}^{n} \subset \textnormal{Herm}(\mathbb{C}^{d_{\rm A}} \otimes \mathbb{C}^{d_{\rm B}})$ is called a \textbf{non-adaptive} \textnormal{LOCC} \textnormal{(NA-LOCC)} measurement if it can be decomposed as 
    \begin{equation}
        \label{eq:na_deco}
        M^{\lambda} = \sum_{a=1}^{m} \sum_{b=1}^{m_{\rm B}} p(\lambda|a,b) \, A^{a} \otimes B^{b},
    \end{equation}
    where $\{A^{a}\}_{a=1}^{m}$ and $\{B^{b}\}_{b=1}^{m_{\rm B}}$ are local \textnormal{POVM's} in Alice's and Bob's systems, respectively, and $\{p(\lambda|a,b)\}_{\lambda = 1}^{n}$ is a probability distribution for each $(a, b)  \in \{1,\dots,m\} \times \{1,\dots,m_{\rm B}\}$.

We define $\mathrm{NA\text{-}LOCC}$ to be the set of all bipartite measurements $\{M^\lambda\}_{\lambda=1}^n$, for any number of outcomes $n$, for which there exist $m,m_{\rm B}\in\mathbb{N}$, local \textnormal{POVMs} $\{A^{a}\}_{a=1}^{m}$ and $\{B^{b}\}_{b=1}^{m_{\rm B}}$, and conditional probabilities $p(\lambda|a,b)$ such that \eqref{eq:na_deco} holds.

For a fixed $m\in\mathbb{N}$, we define $\mathrm{NA\text{-}LOCC}_{m}\subseteq \mathrm{NA\text{-}LOCC}$ to be the subset of measurements $\{M^\lambda\}_{\lambda=1}^n$ for which there exists $m_{\rm B}\in\mathbb{N}$, local \textnormal{POVMs} $\{A^{a}\}_{a=1}^{m}$ and $\{B^{b}\}_{b=1}^{m_{\rm B}}$, and conditional probabilities $p(\lambda|a,b)$ such that \eqref{eq:na_deco} holds.
\end{Definition}
One could likewise consider the refined class $\mathrm{NA\text{-}LOCC}_{m,m_{\rm B}}$, where both Alice's and Bob's local POVMs have bounded numbers of outcomes. However, since the technique we present below
is designed to bound the number of outcomes on Alice's side only we focus only on $\mathrm{NA\text{-}LOCC}_{m}$. As we will see in Sec.~\ref{sec:double_trine}, this is enough to show an advantage of adaptive strategies.
The operational difference between 1R-LOCC and NA-LOCC protocols is presented in the diagrams in Figure~\ref{fig:locc-comparison}. 
\begin{figure}
  \centering
  \begin{subfigure}[t]{0.48\textwidth}
    \centering
    \resizebox{\linewidth}{!}{\begin{tikzpicture}

    \node[quantum_box] (A_box) at (0, 0) {$\{A^{a}\}$};

    \node[quantum_box, right=3cm of A_box] (B_box) {$\{B^{\lambda|a}\}$};

    \node[quantum_state] (Psi_state) at ($ (A_box.east)!0.5!(B_box.west) + (0, 1.2) $) {$\psi$};


    \draw[quantum_line]  (Psi_state.west) -- (A_box.north east);

    \draw[quantum_line] (Psi_state.east) -- (B_box.north west);

    \draw[classical_line] (A_box.east) -- (B_box.west) node[label_style, midway, above] {$a \in \{1,\dots,m\}$};


    \draw[classical_line] (B_box.south) -- ($(B_box.south) + (0, -0.7)$) node[label_style, below=-1mm] {$\lambda$};

\end{tikzpicture}}
    \caption{One-round LOCC protocol}
    \label{fig: 1W-LOCC illustration}
  \end{subfigure}\hfill
  \begin{subfigure}[t]{0.48\textwidth}
    \centering
    \resizebox{\linewidth}{!}{\begin{tikzpicture}

    \node[quantum_box] (A_box) at (0, 0) {$\{A^{a}\}$};

    \node[quantum_box, right=3cm of A_box] (B_box) {$\{B^{b}\}$};

    \node[classical_box, below=0.7cm of B_box] (C_box) {$p(\lambda|a,b)$};

    \node[quantum_state] (Psi_state) at ($ (A_box.east)!0.5!(B_box.west) + (0, 1.2) $) {$\psi$};


    \draw[quantum_line]  (Psi_state.west) -- (A_box.north east);

    \draw[quantum_line] (Psi_state.east) -- (B_box.north west);

\draw[classical_line] (A_box.south) |- (C_box.west) node[label_style, pos=0.25, right] {$a \in \{1,\dots,m\}$};


    \draw[classical_line] (B_box.south) -- (C_box.north) node[label_style, midway, right] {$b$};

        \draw[classical_line] (C_box.south) -- ($(C_box.south) + (0, -0.7)$) node[label_style, below=-1mm]  {$\lambda$};

\end{tikzpicture}}
    \caption{Non-adaptive LOCC protocol}
    \label{fig: NA-LOCC illustration}
  \end{subfigure}

  \caption{Schematic comparison between (a) one-round LOCC protocols and (b) non-adaptive LOCC protocols with message size bounded by $m$. In one-round LOCC measurements, the measurement outcome $a$ of the first party can influence the measurement setting of the other party. In non-adaptive protocols, both the local measurements are independently performed before the individual outcomes are post-processed classically.}
  \label{fig:locc-comparison}
\end{figure}
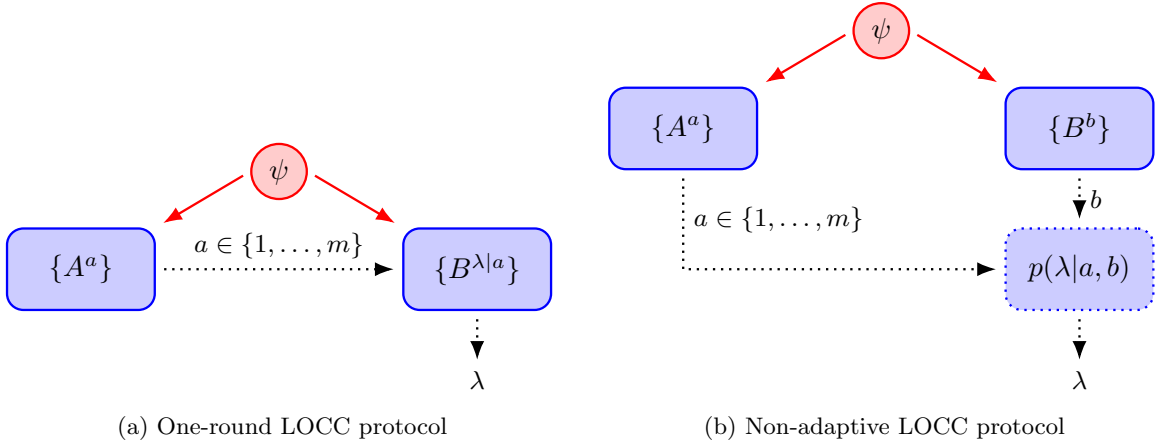

There is a tight connection between non-adaptive \text{LOCC} measurements and the concept of joint measurability that explains how such non-adaptive measurements can be understood as a particular subset of one-round LOCC measurements. To see this, let us rewrite the non-adaptive measurement elements in Eq.~\eqref{eq:na_deco} as 
\begin{equation}
    M^{\lambda} = \sum_{a=1}^{m}  A^{a} \otimes \left( \sum_{b=1}^{m_{\rm B}} p(\lambda|a,b) B^{b}\right) \equiv \sum_{a=1}^{m}  A^{a} \otimes \Tilde{B}^{\lambda|a}.
\end{equation}
It can be verified directly that $\{\Tilde{B}^{\lambda|a}\}_{\lambda=1}^{n}$ with
\begin{equation}
    \label{eq:class_pp}
    \Tilde{B}^{\lambda|a} = \sum_{b=1}^{m_{\rm B}}p(\lambda|a,b) B^{b}
\end{equation}
is a valid \text{POVM} for each value $a$.
In this way, we pass from the intermediate outcome index $b\in\{1,\dots,m_{\rm B}\}$ of the parent POVM
$\{B^{b}\}_{b=1}^{m_{\rm B}}$ to the final outcome index $\lambda\in\{1,\dots,n\}$ by absorbing the classical
post-processing $p(\lambda|a,b)$ into Bob’s operators.
Furthermore, the measurements $\{\Tilde{B}^{\lambda|a}\}_{\lambda=1}^{n}$ running over the different settings $a \in \{1,\dots,m\}$ are \textit{jointly measurable} in the sense that their respective statistics can be inferred from performing a single measurement $\{B^{b}\}_{b=1}^{\rm m_{\rm B}}$ and a subsequent classical post-processing described by Eq.~\eqref{eq:class_pp}, see Ref.~\cite{guehne2023} for a review on joint measurability. This can be summarised by the
following proposition:

\begin{Proposition}
    The sets of one-round \textnormal{LOCC} measurements and non-adaptive \textnormal{LOCC} measurements are related as 
    \begin{equation}
         \textnormal{NA-LOCC}_{m} \subseteq \textnormal{1R-LOCC}_{m}. 
    \end{equation}
    In particular, if $\{M^{\lambda}\}_{\lambda = 1}^{n} \in \textnormal{NA-LOCC}_{m}$ there is a decomposition
    \begin{equation}
        \label{eq:1r_joint_meas}
        M^{\lambda} = \sum_{a=1}^{m}  A^{a} \otimes B^{\lambda|a}
    \end{equation}
    such that the set of measurements $\{B^{\lambda|a}\}_{\lambda = 1}^{n}$ indexed by the setting $a \in \{1,\dots,m\}$ is jointly measurable. 

\end{Proposition}

There is a definition of joint measurability that is equivalent to the classical post-processing description in Eq.~\eqref{eq:class_pp}. In fact, as it has been shown in \cite{compatibility}, a set of measurements $\{B^{\lambda|a}\}_{\lambda = 1}^{n}$ labelled by $a \in \{1,\dots,m\}$ is jointly measurable if there exists a parent POVM given by an array of positive semidefinite operators $S^{\bm{b}} \in \textnormal{Herm}(\mathbb{C}^{d_{\rm B}})$ with $\bm{b} = (b_1, \dots, b_{m}) \in \{1,\dots,n\}^{\times m}$ that satisfy
\begin{align}
    \sum_{b_1,\dots,b_{m} = 1}^{n} S^{\bm{b}} = \mathds{1}_{\rm B} 
\end{align}
and 
\begin{equation}
    B^{\lambda|a} = \sum_{b_1,\dots,b_{m} = 1}^{n} \delta_{b_{a}, \lambda} \, S^{\bm{b}}   
\end{equation}
where $b_a$ denotes the $a$'th entry of the tuple $\bm{b}$. In analogy to the discussion in Section \ref{sec:1r_locc}, one can relax the definition of non-adaptive LOCC measurements by introducing operators $R^{\bm{a}\bm{b}} \in \textnormal{Herm}(\mathbb{C}^{d_{\rm A}^{k}} \otimes \mathbb{C}^{d_{\rm B}})$ of the form 
\begin{equation}
    R^{\bm{a}\bm{b}} = A^{a_1} \otimes \dots \otimes A^{a_{k}} \otimes S^{\bm{b}}
\end{equation}
from which one can extract the non-adaptive measurement $M^{\lambda}$ defined in Eq.~\eqref{eq:1r_joint_meas} via
\begin{equation}
    M^{\lambda} = \frac{1}{d_{\rm A}^{k-1}} \sum_{a_1, \dots, a_k = 1}^{m} \ \sum_{b_1,\dots,b_{m} = 1}^{n} \delta_{b_{a_{k}}, \lambda} \, \Tr_{\rm{A}_1 \dots \rm{A}_{k-1}}(R^{\bm{a} \bm{b}}).
\end{equation}

By demanding that $R^{\bm{a} \bm{b}}$ fulfils suitable collection of affine-linear constraints, we can formulate the following characterisation of the convex hull of non-adaptive LOCC measurements. 

\begin{theorem}[$\textnormal{NA}-\textnormal{LOCC}_{m}$ hierarchy]
    \label{thm:na_locc_m_fixed_char}
    A bipartite measurement $\{M^{\lambda}\}_{\lambda = 1}^{n} \subset \textnormal{Herm}(\mathbb{C}^{d_{\rm A}} \otimes \mathbb{C}^{d_{\rm B}})$ lies in the convex hull of $\textnormal{NA-LOCC}_{m}$ if and only if for each $k \in \mathbb{N}$ there exists an array of positive semidefinite operators $R^{\bm{a} \bm{b}} \in \textnormal{Herm}(\mathbb{C}^{d_{\rm A}^{k}} \otimes \mathbb{C}^{d_{\rm B}})$, for $\bm{a} = (a_1,\dots, a_k) \in \{1,\dots,m\}^{\times k}$ and $\bm{b} = (b_1, \dots, b_{m}) \in \{1,\dots,n\}^{\times m}$, that satisfies the following properties:
    \begin{align}
        &M^{\lambda} = \frac{1}{d_{\rm A}^{k-1}} \sum_{a_1, \dots, a_k = 1}^{m} \ \sum_{b_1,\dots,b_{m} = 1}^{n} \delta_{b_{a_{k}}, \lambda} \, \Tr_{\rm{A}_1 \dots \rm{A}_{k-1}}(R^{\bm{a} \bm{b}}) \textit{ for all } \lambda \label{eq:na_eff_red}\\
        &\sum_{a_1=1}^{m} R^{a_1\dots a_k \bm{b}} = \frac{\mathds{1}_{\rm A}}{d_{\rm A}} \otimes \sum_{a_1 = 1}^{m}\Tr_{\rm A_{1}}( R^{a_1\dots a_k \bm{b}}) \textit{ for all } a_2,\dots, a_k, \bm{b} \\
        & \sum_{b_1, \dots, b_m = 1}^{n} R^{\bm{a}\bm{b}} = \sum_{b_1, \dots, b_m = 1}^{n} \Tr_{\rm B}(R^{\bm{a}\bm{b}}) \otimes \frac{\mathds{1}_{\rm B}}{d_{\rm B}} \textit{ for all } \bm{a} \\
        & (U_{k,d_{\rm A}}^{\sigma} \otimes \mathds{1}) \,  R^{\bm{a} \bm{b}} \, (U_{k,d_{\rm A}}^{{\sigma}^{-1}} \otimes \mathds{1}_{B}) = R^{\sigma(\bm{a}) \bm{b}} \textit{ for all } \sigma, \bm{a}, \bm{b} \\
        & {(R^{\bm{a}\bm{b}})}^{T_{\rm A^{\ell}}} \succcurlyeq 0 \textit{ for all } \bm{a}, \bm{b}, \ell \in \{0,\dots,k\} \\
        &\sum_{a_1 \dots a_k=1}^{m} \sum_{b_{1} \dots b_{m} = 1}^{n}\Tr(R^{\bm{a} \bm{b}}) = d_{\rm A}^k d_{\rm B}. \label{eq:na_norm}
    \end{align}
\end{theorem}
\begin{proof}
    The proof is presented in Appendix \ref{app:thm_2_proof}.
\end{proof}
For increasing values of $k$ the set of measurements $\{M^{\lambda}\}_{\lambda = 1}^{n}$ that satisfy Eq.~\eqref{eq:na_eff_red}-\eqref{eq:na_norm} form a sequence of outer approximations of the non-adaptive measurements $\textnormal{conv}(\textnormal{NA-LOCC}_{m})$ that is guaranteed to converge in the limit $k \rightarrow \infty$. 

We now exhibit the analogous first level $(k=1)$ of the hierarchy for minimum error discrimination with $\text{NA-LOCC}_m$ measurements, which is given in the optimisation problem~\ref{algo: lo-sdp}.

\begin{algorithm}
   $$\begin{array}{ll}
\text{given} & \{(p_\lambda, \varrho_\lambda)\}_{\lambda=1}^n, m \\    
\text{maximize} & \sum_{\lambda=1}^n p_\lambda \text{Tr}(\varrho_\lambda M^\lambda) \\
\text{subject to} &  M^\lambda \coloneqq  \sum_{a=1}^{m} \sum_{b_{1},\dots ,b_{m}=1}^{n} \delta_{b_{a}, \lambda} \, R^{a \bm{b}}\\
&\sum_{a=1}^{m} R^{a \bm{b}}=\mathds{1}_{\rm A}/d_{\rm A} \otimes \sum_{a=1}^{m} \Tr_{\rm A}(R^{a \bm{b}})\quad \forall \bm{b} = (b_1,\dots, b_m)\\
& \sum_{b_1,\dots b_m=1}^n  R^{a \bm{b}}= \sum_{b_1,\dots b_m=1}^n\Tr_{\rm B}(R^{a \bm{b}})\otimes \mathds{1}_{\rm B}/d_{\rm B}\quad \forall a\\
& R^{a \bm{b}}   \succcurlyeq  0\quad \forall a, \bm{b}\\
& (R^{a \bm{b}})^{T_{\rm A}} \succcurlyeq 0 \quad\forall a, \bm{b}
\end{array}$$
\caption{$\text{NA-LOCC}_m$ SDP for minimum-error discrimination at $k=1$}
\label{algo: lo-sdp}
\end{algorithm}

By the same extremality argument used in the previous section, the optimum discrimination probability over $\text{NA-LOCC}_m$ saturates once $m \ge d_A^2$. Consequently, sweeping $m$ over the range $1 < m \le d_A^2$ provides an optimal value for the entire NA-LOCC set. For any $k \in \mathbb{N}$, Theorem \ref{thm:na_locc_m_fixed_char} can be used to compute an upper bound for the state discrimination success probability using non-adaptive strategies. Therefore, if any adaptive strategy gives rise to a higher value than what is obtained by such an upper bound, one certifies an advantage of adaptive measurement strategies over all non-adaptive strategies in the sense of question (d) in the Introduction; see Section \ref{sec:double_trine} for a specific example. 

\section{Examples}
\label{sec:examples}
We now present several applications of our hierarchies to distinct state discrimination tasks and show how they enable a fine-grained analysis regarding the operational LOCC parameters that would not be possible with previous methods. All computations were performed with a Julia implementation, available in our GitHub repository~\cite{github_repo}.

\subsection{Iso-entangled bases}

In our first example, we show how our technique allows us to reveal the directional dependence of LOCC protocols in the sense that for some state discrimination tasks there is a preferred direction of communication. 
Furthermore, the example is designed to showcase how entanglement in the state ensemble affects the optimal performance in their discrimination when restricting to LOCC measurements. To approach this question, we consider the so-called iso-entangled two-qubit bases that have been studied in Ref.~\cite{isoentangled_basis}, consisting of complete sets of orthogonal pure states in which every state has the same entanglement; the latter is quantified by the entanglement measure called \textit{tangle} computed as $T(\ket{\psi}) = 2(1-\Tr[(\Tr_{\rm B}(\ketbra{\psi}))^2])$, i.e., the square of the concurrence \cite{concurrence}.  
We focus on the so-called Bell-basis family that comprises special instances of iso-entangled bases, from which we extract a one-parameter family of bases ranging from a separable basis to the maximally entangled basis of Bell states. 
\begin{example}[Bell-basis family]
The Bell-basis family is the family of two-qubit bases parametrized as
\begin{align*}
&\ket{\psi_1} = \sin(\delta)\,\ket{01} \;+\; \sin(\tau)\,\cos(\delta)\,\ket{10} \;-\; \cos(\tau)\,\cos(\delta)\,\ket{11}
\\
&\ket{\psi_2} = \cos(\delta)\,\ket{01} \;-\; \sin(\tau)\,\sin(\delta)\,\ket{10} \;+\; \cos(\tau)\,\sin(\delta)\,\ket{11} \\
&\ket{\psi_3} = -\,\cos(\delta)\,e^{i\xi}\,\ket{00} \;+\; \cos(\tau)\,\sin(\delta)\,\ket{10} \;+\; \sin(\tau)\,\sin(\delta)\,\ket{11}
\\
&\ket{\psi_4} = \sin(\delta)\,e^{i\xi}\,\ket{00} \;+\; \cos(\tau)\,\cos(\delta)\,\ket{10} \;+\; \sin(\tau)\,\cos(\delta)\,\ket{11},
\end{align*}
all of which have the following tangle
\begin{equation}
    T(\delta, \tau, \xi) = \sin^2(2\delta)\sin^2(\tau).
\end{equation}
\end{example}
It should be noted that the states $\ket{\psi_{\lambda}}$ form an orthonormal basis for each parameter setting so that they can be perfectly distinguished by collective measurements. Therefore, we study the effect of entanglement in local discrimination strategies, where we will see that entanglement prevents perfect discrimination.  

We restrict our attention to the case where \(\xi=\pi/2\) and \(\delta=\pi/4\), and study the optimal 1R-LOCC success probability as a function of \(\tau\), with particular attention to the question (c) of the Introduction on the preferred direction of communication. Using the $\text{1R-LOCC}_m$ SDP hierarchy up to the level $k=2$ provided by Theorem \ref{thm:1W_locc_m_fixed_char}, we compute upper bounds to the optimal success probability of discrimination for both possible directions of communication (Alice measures first or Bob measures first) with message size \((m=2)\), see Figure \ref{fig: isoentangled}. The latter can be obtained using the same optimisation problem but swapping the qubits of Alice and Bob.  
It turns out that the bounds obtained at this level of the hierarchy are actually exact up to the numerical accuracy, as they coincide with lower bounds that are obtained by performing a see-saw optimisation routine. The latter just describes the optimisation technique in which either Alice's or Bob's measurements are fixed, and the optimisation is performed regarding the measurements of the other party. By extracting the optimising measurement, this step can be repeated iteratively until convergence in the success probability is obtained.

\begin{figure}
    \centering
    \includegraphics[width=0.7\linewidth]{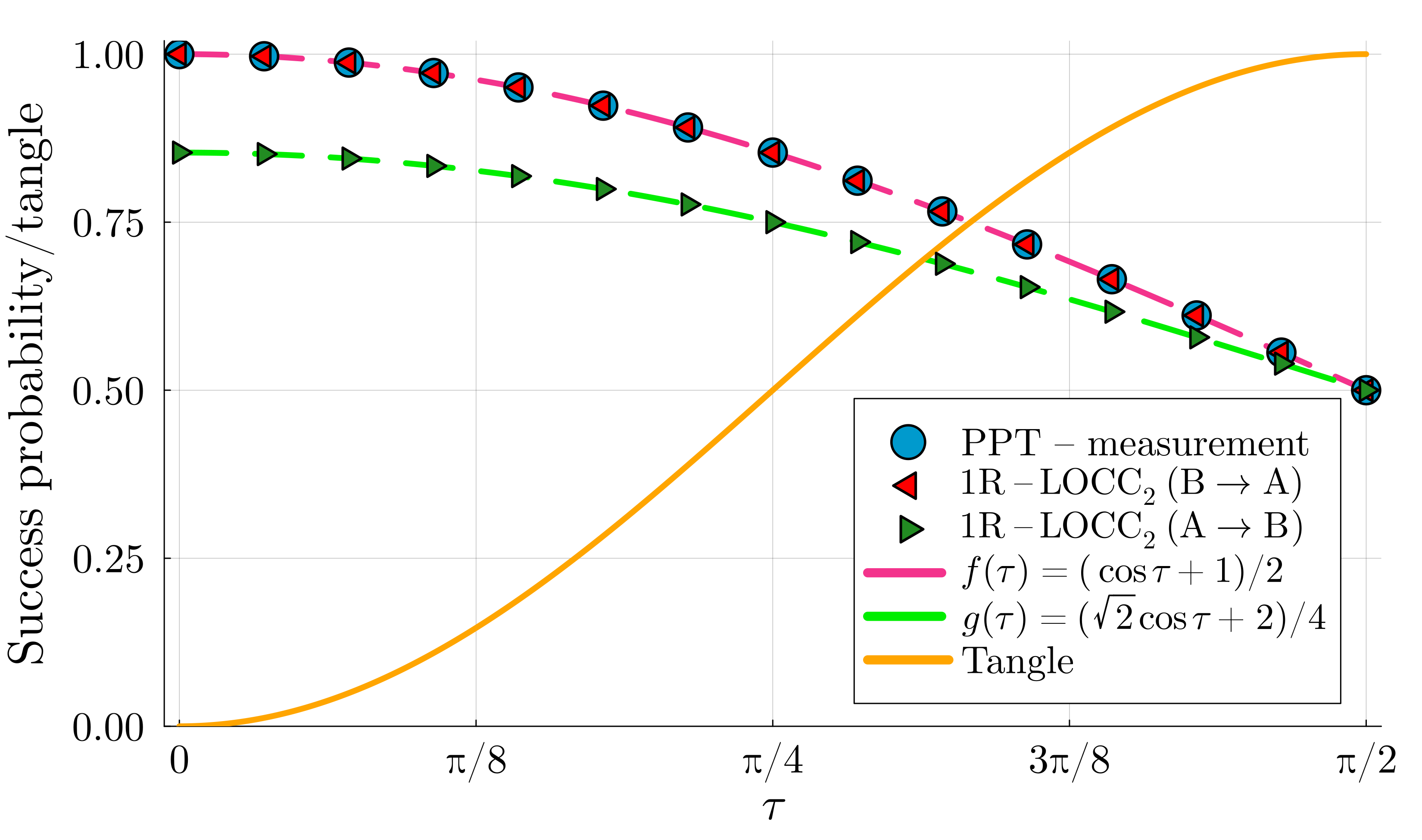}
    \caption{Optimal state discrimination of the Bell-basis family using Alice-first and Bob-first LOCC measurements with one bit of communication, together with the upper bound from the PPT relaxation. It can be seen that there is a preferred direction of communication (Bob $\rightarrow$ Alice). For a comparison between the optimal success probabilities and the entanglement present in the state ensemble, we also plotted the tangle $T$ of the respective basis vectors.}
    \label{fig: isoentangled}
 \end{figure}

The values depicted in the plot of Figure~\ref{fig: isoentangled} show that, for this family of bases, one-round LOCC with Alice measuring first, denoted $\text{LOCC}\,{\rm A\to B}$, consistently outperforms the reverse order $\text{LOCC} \,{\rm B\to A}$ with a single bit of communication ($m=2$). We compared our values with the relaxation that all measurement operators are PPT and observed that such values coincide with our results for the $\rm B\to A$ communication. Since PPT measurements form an outer approximation of all LOCC measurements (unlimited rounds and messages of unlimited size), we deduced that in this case the optimal LOCC strategy only requires one round of classical communication with a single bit where Bob has to measure first. Therefore, the example gives a particularly interesting answer to question (a) asked in the introduction. To obtain the accurate value for the disadvantageous direction of communication $\rm{A} \to \rm{B}$ the level $k=2$ has been necessary. 
It should be noted that it is not possible to determine the preferable direction of communication when considering the PPT relaxation alone.

To build an intuition, we will now analyse the strategies for distinguishing between the elements of the Bell-basis family that accurately produce the optimal success probabilities of discrimination displayed in Figure~\ref{fig: isoentangled} obtained from our numerics.

\subsubsection{Explicit Alice → Bob strategy}
If Alice measures first, we will see that Bob can simply use a fixed measurement that does not depend on Alice's outcome. It is hence a non-adaptive LOCC measurement; see Section \ref{sec:na_meas}. 
The optimal strategy starts with Alice measuring the two-outcome observable
\begin{equation}
\label{eq:A-obs}
A \;=\; \tfrac{\sqrt{2}}{2}\,(\sigma_x+\sigma_y)
\end{equation}
that has outcomes $a = \pm 1$. This measurement choice is independent of the value of the parameter $\tau$.

If Alice observes the outcome $+1$, she guesses that the state is either $\ket{\psi_3}$ or $\ket{\psi_2}$; otherwise if she observes the outcome $-1$, she supposes that the state is either $\ket{\psi_4}$ or $\ket{\psi_1}$.
 The outcome probabilities for each state group are
\begin{equation}
\label{eq:A-probs}
\begin{array}{ll}
p\big(+\,\big|\,\psi_1 \text{ or } \psi_4\big)=\dfrac{-\sqrt{2}\cos\tau+2}{4},
&
p\big(+\,\big|\,\psi_2 \text{ or } \psi_3\big)=\dfrac{\sqrt{2}\cos\tau+2}{4},\\[8pt]
p\big(-\,\big|\,\psi_1 \text{ or } \psi_4\big)=\dfrac{\sqrt{2}\cos\tau+2}{4},
&
p\big(-\,\big|\,\psi_2 \text{ or } \psi_3\big)=\dfrac{-\sqrt{2}\cos\tau+2}{4},
\end{array}
\end{equation}
from which the error probability of Alice's guess can be deduced as 
\begin{equation}
    p_{\rm err} = \dfrac{-\sqrt{2}\cos\tau+2}{4}. 
\end{equation}
After Alice's measurement, regardless of her outcome, Bob's possible reduced post-measurement states are pairwise orthogonal, i.e., 
\begin{equation}
\label{eq:Bob-orth-Afirst}
\sigma_B^1 \perp \sigma_B^4
\quad\text{and}\quad
\sigma_B^2 \perp \sigma_B^3,
\end{equation}
so if Alice is correct in her assessment between $\{\psi_2,\psi_3\}$ vs.~$\{\psi_1,\psi_4\}$, Bob can perfectly discriminate within the pair using a single fixed observable
\begin{equation}
\label{eq:B-fixed-Afirst}
B \;=\; 
\frac{(2\cos\tau+\sqrt{2})\sin\tau}{\sqrt{2}\cos\tau+2}\,\sigma_x
\;+\;
\frac{\sqrt{2}\sin\tau}{\sqrt{2}\cos\tau+2}\,\sigma_y
\;+\;
\frac{(2\cos\tau+\sqrt{2})\sin\tau}{\sqrt{2}\cos\tau+2}\,\sigma_z.
\end{equation}

Since Bob's measurement outcome can be predicted deterministically, the total success probability equals the probability that Alice's outcome corresponds to the correct pair. Hence, under the uniform prior, this gives
\begin{equation}
\label{eq:Psucc-Afirst}
p^{\text{succ}}_{A\to B}(\tau)
= \frac{1}{4}\Big[
p(-|\psi_1)+p(-|\psi_4)+p(+|\psi_2)+p(+|\psi_3)
\Big]
= \frac{\sqrt{2}\cos\tau+2}{4},
\end{equation}
which we plot in Figure~\ref{fig: isoentangled} to show that it coincides with the numerical results and saturates the upper bounds up to numerical accuracy.
\subsubsection{Explicit Bob → Alice strategy}

For optimal discrimination, given a value of $\tau$, Bob measures the two-outcome observable
\begin{equation}
\label{eq:B-obs}
B \;=\; \sin\tau\,\sigma_x+\cos\tau\,\sigma_z,
\end{equation}
that has outcomes $b = \pm 1$. 

If Bob observes the outcome \(+\), he assumes the state is either \(\psi_3\) or \(\psi_4\); if he observes the outcome \(-\), he restricts it to \(\psi_1\) or \(\psi_2\).
 The outcome probabilities for each state group are\begin{equation}
\label{eq:B-probs}
\begin{array}{llll}
p(+\,|\,\psi_1\text{ or }\psi_2)=\dfrac{1-\cos\tau}{2},
&
p(+\,|\,\psi_3\text{ or }\psi_4)=\dfrac{1+\cos\tau}{2},\\[8pt]
p(-\,|\,\psi_1\text{ or }\psi_2)=\dfrac{1+\cos\tau}{2},
&
p(-\,|\,\psi_3\text{ or }\psi_4)=\dfrac{1-\cos\tau}{2}.
\end{array}
\end{equation}
After Bob's measurement, regardless of his outcome, Alice's possible partial states are pairwise orthogonal;
\begin{equation}
\label{eq:Alice-orth-Bfirst}
\sigma_A^1 \perp \sigma_A^2
\quad\text{and}\quad
\sigma_A^3 \perp \sigma_A^4,
\end{equation}
so, if Bob is correct in his assessment between $\{\psi_3,\psi_4\}$ vs $\{\psi_1,\psi_2\}$, Alice can perfectly discriminate within the pair with one of these observables, depending on Bob's outcome:
\begin{equation}
\label{eq:A-adapt}
A^{+}=\sigma_x,
\qquad
A^{-}=\sigma_y.
\end{equation}

Since Alice's outcome can be predicted deterministically once the correct pair is communicated, the total success probability equals the probability that Bob's outcome corresponds to the correct pair. Under the uniform prior, this gives
\begin{equation}
\label{eq:Psucc-Bfirst}
p^{\text{succ}}_{B\to A}(\tau)
= \frac{1}{4}\Big[
p(+|\psi_1)+p(+|\psi_2)+p(-|\psi_3)+p(-|\psi_4)
\Big]
= \frac{1+\cos\tau}{2},
\end{equation}
which we also plot in Figure~\ref{fig: isoentangled} to show that it accurately matches the numerical results.

\subsection{Double trine}
\label{sec:double_trine}
Most results on state discrimination via LOCC address \emph{perfect} protocols, where the goal is to certify that a set of states can be distinguished with probability 1, rather than to maximize the minimum-error success probability. In the literature, either there are derived analytic criteria \cite{hardy_bipartite,2x2_distinguishability,3x3_distinguishability,2x3_distinguishability} or upper bounds via the PPT relaxation  \cite{ppt_relax_1,ppt_relax_2, discrimination_thiesis} are provided. 
With our method presented in this work, we can now ask which conclusions from the perfect setting persist in the probabilistic (minimum-error) setting.

A previously known result is that, for perfect discrimination in \(n\)-qubit systems, allowing general local measurements gives no advantage over local projective measurements. In \cite{LPCC_enough}, it was found that for an $n$-qubit system, a set of mutually orthogonal states can be perfectly distinguished by LOCC measurements if and only if this discrimination can be achieved with local projective measurements.

Using our \(\text{1R-LOCC}_m\) SDP hierarchy provided in Theorem \ref{thm:1W_locc_m_fixed_char}, we can show that this result does not generalise to the minimum error discrimination problem. Hence, non-projective measurements can provide an advantage over all projective measurements in the success probability of discrimination. One basic example where we can see this is the two-qubit “double trine” ensemble, defined as a uniform distribution over the product states below.
\begin{example}[Double trine \cite{double_trine}] The double trine is an equiprobable ensemble of the states 
$\ket{\psi_i}=\ket{s_i}\otimes\ket{s_i}$ for $i=0,1,2$, where
$$
\ket{s_0}=\ket{0},
\qquad
\ket{s_1}=-\frac12\ket{0}-\frac{\sqrt3}{2}\ket{1},
\qquad
\ket{s_2}=-\frac12\ket{0}+\frac{\sqrt3}{2}\ket{1}.
$$
\end{example}

To test whether projective measurements suffice, we will search for changes in the optimal discrimination probabilities for different message budgets \(m\). For qubits, any local projective measurement has two outcomes, so every protocol with projective measurements is realizable within \(m=2\). Hence, if the optimal one-round LOCC success probability increases to \(m>2\), this certifies a benefit from using non-projective POVMs. We evaluate this on the double trine ensemble. For each \(m\in\{2,3,4\}\), we compute a certified upper bound on the one round LOCC optimum using the \(\text{1R-LOCC}_m\) SDP hierarchy at level \(k=3\), and we obtain a lower bound via a randomized see-saw heuristic over \(\text{1R-LOCC}_m\) strategies. The resulting bounds are summarised in Table~\ref{tab: trine}.

\begin{table}
  \centering
\begin{tabular}{ccccc}
  \toprule
  \textbf{$m$} & \textbf{\shortstack{\(\text{1R-LOCC}_m\)\\see-saw}}  & \textbf{\shortstack{\(\text{1R-LOCC}_m\)\\SDP hierarchy}} & \textbf{\shortstack{\(\text{NA-LOCC}_m\)\\see-saw}} & \textbf{\shortstack{\(\text{NA-LOCC}_m\)\\SDP hierarchy}} \\
  \midrule
  2 & 0.8976 & 0.905  & 0.8003 & 0.8116\\
  3 & 0.9330  & 0.9346 & 0.8079 & 0.8248\\
  4 & 0.9330 & 0.950 & 0.8079 & 0.8509 \\
  \bottomrule
\end{tabular}
    \caption{Bounds on minimum error success probabilities for the double trine under \(\text{1R-LOCC}_m\) as a function of the message alphabet size \(m\); “see-saw” gives lower bounds, “SDP hierarchy” gives upper bounds; \(\text{NA-LOCC}_m\) denotes the non-adaptive subclass. The most important comparison is between to leftmost and rightmost columns, providing lower bounds achieved by adaptive strategies and upper bounds achieved by non-adaptive strategies, respectively. One can see that for each number of outcomes $m$ on Alice's side, there is a clear gap between adaptive and non-adaptive strategies.}
    \label{tab: trine}
\end{table}

As seen from Table~\ref{tab: trine}, the upper bound increases with the message budget \(m\). In particular, the \(m=3\) lower bound from the \(\text{1R-LOCC}_m\) see–saw, \(0.9330\), exceeds the \(m=2\) upper bound from the \(\text{1R-LOCC}_m\) SDP hierarchy, \(0.905\). As discussed above, this certifies that projective measurements are not sufficient to reach the one–round LOCC optimum in the minimum–error setting. Equivalently, non-projective POVMs strictly improve the performance on this ensemble beyond what projective strategies can achieve.

In this example, we can also demonstrate the necessity of adaptivity in Bob's POVM.  For each \(m\), we compute certified upper bounds on the non-adaptive LOCC optimum using the \(\text{NA-LOCC}_m\) SDP hierarchy (level \(k=3\)), and we pair these with lower bounds obtained via a randomized see-saw heuristic restricted to non-adaptive strategies. The resulting bounds are also reported in Table~\ref{tab: trine}.

 Even allowing arbitrary single-qubit POVMs for the first measuring party, which can be indexed with two classical bits of communication, the non-adaptive upper bound to the success probability for \(m=4\) from the \(\text{NA-LOCC}_m\) SDP hierarchy is \(0.8509\). Any extremal POVM on a qubit has at most four outcomes, so $m=4$ suffices to label the local outcome. This already lies below the adaptive one round LOCC performance at \(m=2\), whose lower bound from the \(\text{1R-LOCC}_m\) see-saw is \(0.8976\). Hence, for this ensemble, achieving the best discrimination probabilities requires adaptivity in addition to non-projective measurements.

 To conclude this example, the discrimination of the double trine ensemble demonstrates several interesting aspects. First, we see that, even in the case of qubits, the usage of non-projective measurements for the sending party can provide an advantage over projective measurements. Second, the ensemble provides an example for the certification of the advantage of adaptive LOCC protocols over all non-adaptive protocols in the sense of question (d) asked in the Introduction. 

\subsection{Two ququarts}
We can also obtain results from our method in higher-dimensional systems. Consider the following equiprobable ensemble of maximally entangled ququart–ququart ($d_{\rm A} = d_{\rm B} = 4$) states.

\begin{example}[Maximally entangled ququart-ququart states]
\begin{align}
|\psi_1\rangle &= \frac{1}{2}\bigl(|0\rangle|0\rangle + |1\rangle|1\rangle + |2\rangle|2\rangle + |3\rangle|3\rangle\bigr),\\
|\psi_2\rangle &= \frac{1}{2}\bigl(|0\rangle|3\rangle + |1\rangle|2\rangle + |2\rangle|1\rangle + |3\rangle|0\rangle\bigr),\\
|\psi_3\rangle &= \frac{1}{2}\bigl(|0\rangle|1\rangle + |1\rangle|0\rangle - |2\rangle|3\rangle - |3\rangle|2\rangle\bigr).
\end{align}
\end{example}

Using a PPT relaxation, one can certify that these states can be perfectly discriminated with PPT measurements. In fact, it is easy to see that it can be done with a 1R-LOCC protocol with two bits of communication ($m=4$): Alice and Bob could both perform a computational basis measurement.

The PPT relaxation, however, does not tell us if this strategy is optimal in the sense that there could be a strategy with $m<4$ that also allows for a perfect discrimination. Our SDP hierarchy can tell that this is not the case. In Table~\ref{tab: quart} we show the result of running the \(\text{1R-LOCC}_m\) SDP hierarchy at $k=2$ for the ququart ensemble with different amounts of communication. We see that perfect discrimination requires at least two bits of communication ($m>3$) since the upper bound to the success probability for $m=3$ is $0.9623$. Referring to the question (b) of the Introduction, the example of maximally entangled ququart-ququart states nicely showcases how the discrimination success may depend on the quantitative amount of classical information in the measurement.
\begin{table}
  \centering
\begin{tabular}{ccc}
  \toprule
  \textbf{$m$} & \textbf{\shortstack{\(\text{1R-LOCC}_m\)\\see-saw}}  & \textbf{\shortstack{\(\text{1R-LOCC}_m\)\\SDP hierarchy}}  \\
  \midrule
  2 & 0.6667 & 0.6667  \\
  3 & 0.8333  & 0.9623 \\
  4 & 1 & 1  \\
  \bottomrule
\end{tabular}
   \caption{Optimal \(\text{1R-LOCC}_m\) discrimination probability of the ququart–ququart ensemble, as a function of the message alphabet size \(m\); “see-saw” gives lower bounds, and “SDP hierarchy” gives upper bounds. The values show that the minimal amount of communication required for a perfect discrimination is given by $m=4$.}
    \label{tab: quart}
\end{table}

\section{Conclusion}
\label{sec:conclusion}

We developed a framework to study the amount of classical communication that is needed in an LOCC protocol. Our main technical contribution is a characterisation of the convex hull of one-round LOCC with a bounded message alphabet, \(\text{conv}(1\text{R-LOCC}_m)\), via constrained symmetric extensions. This yields a convergent SDP hierarchy in which the communication budget \(m\), the message direction, and the need for adaptation can be explicitly set. At low hierarchy levels the resulting SDPs are tractable, yet already tight enough to certify optimality in representative instances.

For the iso-entangled Bell-basis family, our method allowed us to detect a directional asymmetry at one bit of communication: Bob-first protocols strictly outperform Alice-first protocols. Using the example of the double trine ensemble, we demonstrated that projective strategies with classical communication are not sufficient to get the optimal success probability and general non-projective measurements provide an advantage.
We also show that our algorithm can be used in systems with higher dimensions. In the ququart-ququart example where a non-adaptive strategy can achieve perfect discrimination with two bits of communication, the hierarchy certifies that perfect performance cannot be achieved with fewer than two classical bits of communication.

Methodologically, our SPD hierarchies provide a general recipe to turn LOCC-type constraints into affine conditions inside semidefinite programs. Beyond state discrimination, the same template applies to tasks such as channel discrimination, quantum state verification, remote state preparation, state distillation, and teleportation, where message size, direction, and number of rounds matter.

There are several natural directions for future work. First, extending the present characterisation to multi-round LOCC with limited per-round message budgets would allow a systematic study of rounds versus bits trade-offs. Second, optimising the implementation of the constraints, particularly regarding the symmetric variables, could further improve the scalability.

\section{Acknowledgements}
First, we would like to express our sincere thanks to Ryszard Horodecki for his pioneering work. In
addition, we would like to thank Marco Túlio Quintino, Roope Uola and Carlos de Gois for insightful discussions. 

This work was supported by the Deutsche Forschungsgemeinschaft (DFG, German Research Foundation, project number 563437167), the 
Sino-German Center for Research Promotion 
(Project M-0294), the German Federal Ministry of Research, Technology and Space (Project QuKuK, Grant No.~16KIS1618K and Project BeRyQC, Grant No.~13N17292), the Project EIN Quantum NRW and the Swedish Research Council (Grant No. 2024-05341).
This work was partially supported by the São 
Paulo Research Foundation (Fapesp) \mbox{[2024/23590-2]} and the National Council for Scientific and Technological Development (CNPq) \mbox{[200705/2025-3]}. We thank the Coaraci Supercomputer for computer time (Fapesp grant 2019/17874-0) and the Center for Computing in Engineering and Sciences at Unicamp (Fapesp grant 2013/08293-7).

\appendix
\section{Symmetric extensions of constrained separable states}
\label{app:sym_ext}
Here, we give a quick review of the concept of constrained separability that is used in the proofs of the Theorems \ref{thm:1W_locc_m_fixed_char} and \ref{thm:na_locc_m_fixed_char} presented in Appendices \ref{ap:theorem_1} and \ref{app:thm_2_proof}, respectively. 
Recall that a positive semidefinite operator $\rho \in \textnormal{Herm}(\mathbb{C}^{{d}_{\rm A}} \otimes \mathbb{C}^{{d}_{\rm B}})$ acting on a bipartite Hilbert space is called separable if it admits a decomposition of the form 
\begin{equation}
    \label{eq:sep_deco}
    \rho = \sum_{r} q_{r} \, \sigma^{r} \otimes \gamma^{r},
\end{equation}
where $\sigma^{r} \in \textnormal{Herm}(\mathbb{C}^{d_{\rm A}})$ and $\gamma^{r} \in \textnormal{Herm}(\mathbb{C}^{d_{\rm B}})$ are positive semidefinite operators and $q_{r}$ are positive numbers. 
Constrained separability, defined as follows, is a requirement one can impose on a positive semidefinite operator that is in general stricter than usual separability.

\begin{Definition}[Constrained separable states]
    \label{def:con_sep}
    Let $d_{\rm A}, d_{\rm B} \in \mathbb{N}$ and let $\Phi:\textnormal{Herm}(\mathbb{C}^{d_{\rm A}}) \rightarrow V_{\rm A}$ and $\Psi:\textnormal{Herm}(\mathbb{C}^{d_{\rm B}}) \rightarrow V_{\rm B}$ be linear maps with real target vector spaces $V_{\rm A}$ and $V_{\rm B}$. Let furthermore $v_{\rm A} \in V_{\rm A}$ and $v_{\rm B} \in V_{\rm B}$ be some constant vectors. 
    Then a positive semidefinite operator $\rho \in \textnormal{Herm}(\mathbb{C}^{{d}_{\rm A}} \otimes \mathbb{C}^{{d}_{\rm B}})$ is called \textit{constrained separable} regarding $(\Phi,v_{\rm A},\Psi, v_{\rm B})$ if it can be decomposed as 
    \begin{equation}
        \rho = \sum_{r} q_{r} \, \sigma^{r} \otimes \gamma^{r},
    \end{equation}
    where $\sigma^{r} \in \textnormal{Herm}(\mathbb{C}^{d_{\rm A}})$ and $\gamma^{r} \in \textnormal{Herm}(\mathbb{C}^{d_{\rm B}})$ are positive semidefinite operators with $\Tr(\sigma^{r}) = \Tr(\gamma^{r}) = 1$ that satisfy
    \begin{align}
        \Phi(\sigma^r) &= v_{\rm A}, \\
        \Psi(\gamma^r) &= v_{\rm B},
    \end{align}
    for all $r$ and $q_{r}$ are positive numbers.
\end{Definition} 

It has been shown in Ref.~\cite{Doherty2004} that separable operators as defined in Eq.~\eqref{eq:sep_deco} have the unique property that they allow the presence of so-called $k$-symmetric extensions for all $k$. There is similar concept also for the case of constrained separable states, or more generally for constrained separable operators if the normalisation $\Tr(\rho) = 1$ is not satisfied. 

\begin{Definition}[Constrained $k$-symmetric extendible states]
    Let $d_{\rm A}, d_{\rm B} \in \mathbb{N}$ and let $\Phi:\textnormal{Herm}(\mathbb{C}^{d_{\rm A}}) \rightarrow V_{\rm A}$ and $\Psi:\textnormal{Herm}(\mathbb{C}^{d_{\rm B}}) \rightarrow V_{\rm B}$ be linear maps with real target vector spaces $V_{\rm A}$ and $V_{\rm B}$. Let furthermore $v_{\rm A} \in V_{\rm A}$ and $v_{\rm B} \in V_{\rm B}$ be some constant vectors. Then a positive semidefinite operator $\rho \in \textnormal{Herm}(\mathbb{C}^{{d}_{\rm A}} \otimes \mathbb{C}^{{d}_{\rm B}})$ is called constrained $k$-symmetric extendible with respect to $(\Phi,v_{\rm A},\Psi, v_{\rm B})$ if there exists a positive semidefinite operator $\bm{\Omega}_{k} \in \textnormal{Herm}({(\mathbb{C}^{d_{\rm A}})}^{\otimes k} \otimes \mathbb{C}^{d_{\rm B}})$ such that
    \begin{align}
        \rho &= \Tr_{\rm A_{1}\dots A_{k-1}}(\bm{\Omega}_{k}), \\
        \bm{\Omega}_{k} &= (U_{k,d_{\rm A}}^{\sigma} \otimes \mathds{1}_{\rm B}) \bm{\Omega}_{k} (U_{k,d_{\rm A}}^{{\sigma}^{-1}} \otimes \mathds{1}_{\rm B}) \textnormal{ for all } \sigma \in S_{k}, \\
        [\Phi \otimes \textnormal{id}_{\rm A_2 \dots A_{k}B}] \bm{\Omega}_{k} &= v_{\rm A} \otimes \Tr_{\rm A_1}(\bm{\Omega}_{k}), \\
        [\textnormal{id}_{\rm A_{1}\dots A_{k}} \otimes \Psi] \bm{\Omega}_{k} &= \Tr_{\rm B}(\bm{\Omega}_{k}) \otimes v_{\rm B}, 
    \end{align}
    where $U_{k,d}^{\sigma}$ is the standard unitary representation of the symmetric group $S_k$ acting on ${(\mathbb{C}^{d})}^{\otimes k}$ and defined by
\begin{equation}
    U_{k,d}^{\sigma} \ket{i_1} \otimes \dots \otimes \ket{i_k} = \ket{i_{\sigma^{-1}(1)}} \otimes \dots \otimes \ket{i_{\sigma^{-1}(k)}}
\end{equation}
for each permutation $\sigma \in S_k$.
\end{Definition}
Another constraint that is often additionally imposed on the extension $\bm{\Omega}_{k}$ for purely practical purposes is the positivity of all partial transposes, i.e.,
\begin{equation}
    \label{eq:abstract_ppt}
    (\bm{\Omega}_{k})^{T_{\rm{A}^{\ell}}} \succcurlyeq 0 \textnormal{ for all } \ell \in \{0,\dots,k\}
\end{equation}
where $A^{\ell}$ denotes the first $\ell$ copies of the Hilbert space $\mathbb{C}^{d_{\rm A}}$. It should be noted that condition \eqref{eq:abstract_ppt} is technically not needed in the definition of constrained $k$-symmetric extendibility.

The question arises how the concepts of constrained separability and constrained symmetric extendibility are related.  
By setting the operator $\bm{\Omega}_k$ to be 
\begin{equation}
    \bm{\Omega}_k = \sum_{r} q_r \, \underbrace{\sigma^{r} \otimes \dots \otimes \sigma^{r}}_{k \textnormal{ times}} \otimes \,\gamma^{r}
\end{equation}
it can be verified that constrained separable states as in Def.~\ref{def:con_sep} are constrained $k$-symmetric extendible. Conversely, in Ref.~\cite{berta_optimization} it has been proven, see also Ref.~\cite{channel_discrimination} for an alternative proof, that also the other implication holds, meaning that a state is constrained separable if and only if it is constrained $k$-symmetric extendible for all $k$. 

\begin{Proposition}[\cite{berta_optimization}]
    \label{prop:consep_convergence}
    Let $d_{\rm A}, d_{\rm B} \in \mathbb{N}$ and let $\Phi:\textnormal{Herm}(\mathbb{C}^{d_{\rm A}}) \rightarrow V_{\rm A}$ and $\Psi:\textnormal{Herm}(\mathbb{C}^{d_{\rm B}}) \rightarrow V_{\rm B}$ be linear maps with real target vector spaces $V_{\rm A}$ and $V_{\rm B}$. Let furthermore $v_{\rm A} \in V_{\rm A}$ and $v_{\rm B} \in V_{\rm B}$ be some constant vectors. Then an operator $\rho \in \textnormal{Herm}(\mathbb{C}^{d_{\rm A}} \otimes \mathbb{C}^{d_{\rm B}})$ is constrained separable with respect to $(\Phi,v_{\rm A},\Psi, v_{\rm B})$ if and only if it is constrained $k$-symmetric extendible for all $k \in \mathbb{N}$. 
\end{Proposition}

Mathematically, Proposition \ref{prop:consep_convergence} is the key ingredient for the proofs of Theorem \ref{thm:1W_locc_m_fixed_char} and \ref{thm:na_locc_m_fixed_char}. 

\section{Proof of Theorem \ref{thm:1W_locc_m_fixed_char}}
\label{ap:theorem_1}
Theorem \ref{thm:1W_locc_m_fixed_char} stated the following.
\begin{theoremcopy}[$1\textnormal{R}-\textnormal{LOCC}_{m}$ hierarchy]
    A bipartite measurement $\{M^{\lambda}\}_{\lambda = 1}^{n} \subset \textnormal{Herm}(\mathbb{C}^{d_{\rm A}} \otimes \mathbb{C}^{d_{\rm B}})$ lies in the convex hull of $1\textnormal{R}-\textnormal{LOCC}_{m}$ if and only if for each level $k \in \mathbb{N}$ there exists an array of positive semidefinite operators $R^{\bm{a}\lambda b} \in \textnormal{Herm}(\mathbb{C}^{d_{\rm A}^{k}} \otimes \mathbb{C}^{d_{\rm B}})$, for $\bm{a} = (a_1,\dots, a_k) \in \{1,\dots,m\}^{\times k}$, $\lambda \in \{1,\dots,n\}$ and $b \in \{1,\dots,m\}$, that satisfies the following constraints.
    \begin{align}
        & M^{\lambda} = \frac{1}{d_{\rm A}^{k-1}}\sum_{a_{1}\dots,a_{k} = 1}^{m} \Tr_{\rm{A}_{1}\dots\rm{A}_{k-1}}(R^{a_{1}\dots a_{k}\lambda a_{k}}) \textit{ for all } \lambda \label{eq:1W_cond_start_copy}\\
        & \sum_{a_1=1}^{m} R^{a_1\dots a_k \lambda b} = \frac{\mathds{1}_{\rm A}}{d_{\rm A}} \otimes \sum_{a_1 = 1}^{m}\Tr_{\rm A_{1}}( R^{a_1\dots a_k \lambda b}) \textit{ for all } a_2,\dots, a_k, \lambda, b \label{eq:left_array_cons_copy}\\
        & \sum_{\lambda = 1}^{n} R^{\bm{a}\lambda b} = \sum_{\lambda=1}^{n} \sum_{b'=1}^{m} \Tr_{\rm B}(R^{\bm{a}\lambda b'}) \otimes \frac{\mathds{1}_{\rm B}}{md_{\rm B}} \textit{ for all } \bm{a},b \label{eq:right_array_cons_copy}\\
        & (U_{k,d_{\rm A}}^{\sigma} \otimes \mathds{1}_{\rm B}) \,  R^{\bm{a}\lambda b} \, (U_{k,d_{\rm A}}^{{\sigma}^{-1}} \otimes \mathds{1}_{\rm B}) = R^{\sigma (\bm{a}) \lambda b} \textit{ for all } \sigma, \bm{a}, \lambda, b \label{eq:array_symmetry_copy}\\
        & {(R^{\bm{a}\lambda b})}^{T_{\rm A^{\ell}}} \succcurlyeq 0 \textit{ for all } \bm{a}, \lambda, b \textit{ and } \ell \in \{0,\dots,k\} \label{eq:array_ppt_copy}\\
        & \sum_{a_1 \dots a_k=1}^{m} \sum_{\lambda = 1}^{n}\Tr( R^{\bm{a}\lambda b}) = d_{\rm A}^k d_{\rm B} \textit{ for all } b. \label{eq:1W_cond_end_copy}
    \end{align}
    In Eq.~\eqref{eq:array_ppt_copy}, $X^{T_{\rm A^{\ell}}}$ denotes the partial transpose applied on the first $\ell$ copies of the system $\rm{A}$ with Hilbert space $\mathbb{C}^{d_{\rm A}}$ and in Eq.~\eqref{eq:array_symmetry_copy}, $\sigma(\bm{a})$ is shorthand notation for the tuple $(a_{\sigma^{-1}(1)},\dots,a_{\sigma^{-1}(k)})$. 
   \end{theoremcopy}
        \begin{proof}[Proof of Theorem \ref{thm:1W_locc_m_fixed_char}]
            First assume that $\{M^{\lambda}\}_{\lambda = 1}^{n} \in \textnormal{conv}(1\textnormal{R}-\textnormal{LOCC}_{m})$. By definition, this means that we can decompose each measurement operator as 
        \begin{equation}
            M^{\lambda} = \sum_{r} \sum_{a = 1}^{m} q_{r} \, A_{r}^{a} \otimes B_{r}^{\lambda|a}
        \end{equation}
        where $\{q_r\}_r$ is some probability distribution, $\{A_{r}^{a}\}_{a=1}^{m}$ is a measurement for each $r$ and $\{B_{r}^{\lambda|a}\}_{\lambda=1}^{n}$ is a measurement for each $r$ and $a$.  It can then be checked directly that the operator $R^{\bm{a} \lambda b}$ given by 
        \begin{equation}
            R^{\bm{a} \lambda b} = \sum_{r} q_{r} \, A_{r}^{a_1} \otimes \dots \otimes A_{r}^{a_{k}} \otimes B_{r}^{\lambda|b}
        \end{equation}
        satisfies the conditions \eqref{eq:1W_cond_start_copy}-\eqref{eq:1W_cond_end_copy}. This completes the first direction of the theorem. 

        To demonstrate the opposite direction, we first reformulate the definition of the set $\textnormal{conv}(1\textnormal{R}-\textnormal{LOCC}_{m})$ in a convenient way. A measurement $\{M^{\lambda}\}_{\lambda = 1}^{n} \in \textnormal{conv}(1\textnormal{R}-\textnormal{LOCC}_{m})$ can be identified by the block-diagonal operator $\bm{M} \in \textnormal{Herm}(\mathbb{C}^{n} \otimes \mathbb{C}^{d_{\rm A}} \otimes \mathbb{C}^{d_{\rm B}})$ given by 
        \begin{align}
            \bm{M} &= \sum_{\lambda = 1}^{n} \ketbra{\lambda}_n \otimes M^{\lambda} \\
            &= \sum_{\lambda = 1}^{n} \ketbra{\lambda}_n \otimes \sum_{r} \sum_{a=1}^{m} q_{r} \, A_{r}^{a} \otimes B_{r}^{\lambda|a} \\
            &= \leftindex_{m,m'}{\bra{\bm{\Phi}^+}} \sum_{r} q_{r} \, \bm{A}^r \otimes \bm{B}^{\rm r} \ket{\bm{\Phi}^+}_{m,m'}
        \end{align}
        where the last step uses the identity $\abs{\bra{\bm{\Phi}^+}\ket{a,b}}^2 = \delta_{ab}$ using the unnormalised maximally entangled state $\ket{\bm{\Phi}^+}_{m,m'} = \sum_{a=1}^{m} \ket{a}_{m} \otimes \ket{a}_{m'}$. The block-diagonal operators $\bm{A}^{r} \in \textnormal{Herm}(\mathbb{C}^{md_{\rm A}})$ and $\bm{B}^{r} \in \textnormal{Herm}(\mathbb{C}^{nmd_{\rm B}})$ are given by
        \begin{align}
            \bm{A}^r &= \sum_{a=1}^{m} \ketbra{a}_{m} \otimes A_{r}^{a} \\
            \bm{B}^r &= \sum_{\lambda = 1}^{n} \sum_{b=1}^{m} \ketbra{\lambda}_{n} \otimes \ketbra{b}_{m'} \otimes B_{r}^{\lambda|b}. 
        \end{align}
        The measurement operators $M^\lambda$ can be obtained from $\bm{M}$ simply via $M^\lambda = \leftindex_{n}{\bra{\lambda}} \bm{M} \ket{\lambda}_{n}$. From this reformulation we can see that the convex hull of one-round LOCC measurements with bounded communication can be identified with particular separable quantum states $\bm{X} \in \textnormal{Herm}(\mathbb{C}^{md_{\rm A}} \otimes \mathbb{C}^{nmd_{\rm B}})$ of the form 
        \begin{equation}
            \label{eq:sep_deco_meas}
            \bm{X} = \sum_{r} q_{r} \, \bm{A}^{r} \otimes \bm{B}^{r}
        \end{equation}
        where $\bm{A}^{r} \in \textnormal{Herm}(\mathbb{C}^{md_{\rm A}})$ and $\bm{B}^{r} \in \textnormal{Herm}(\mathbb{C}^{nmd_{\rm B}})$ are positive semidefinite operators that additionally obey the affine constraints
        \begin{align}
            \Tr_{m}(\bm{A}^{r}) &= \sum_{a} A_{r}^{a} = \mathds{1}_{\rm A} \label{eq:locc_con_A}\\
            \Tr_{n}(\leftindex_{m'}{\bra{b'}} \bm{B}^r \ket{b'}_{m'}) &= \sum_{\lambda = 1}^{n} B_{r}^{\lambda|b'} = \mathds{1}_{\rm B} \textnormal{ for all } b'. \label{eq:locc_con_B}
        \end{align}
        Hence, generally speaking, $\bm{X}$ is a bipartite separable quantum state where the tensor factors in its separable decomposition \eqref{eq:sep_deco_meas} are constrained to obey the additional affine constraints \eqref{eq:locc_con_A} and \eqref{eq:locc_con_B}. In regard of the framework explained in Appendix \ref{app:sym_ext}, we refer to such states as constrained separable states.  

        Importantly, by Proposition \ref{prop:consep_convergence} $\bm{X}$ is a constrained separable state if and only if it is constrained $k$-symmetric extendible, meaning that for each $k \in \mathbb{N}$ there exists a positive semidefinite operator $\bm{\Omega}_{k} \in \textnormal{Herm}(\mathbb{C}^{(md_{\rm A})^k} \otimes \mathbb{C}^{nmd_{\rm B}})$
        that satisfies the following properties.
        \begin{align}
            \bm{\Omega}_k  &= (U_{k,md_{\rm A}}^{\sigma} \otimes \mathds{1}_{nmB})\, \bm{\Omega}_k \, (U_{k,md_{\rm A}}^{\sigma^{-1}}\otimes \mathds{1}_{nmB}) \textnormal{ for all } \sigma \in S_k \label{eq:con_sym_1}\\
            \Tr_{m_1} (\bm{\Omega}_k) &= \frac{\mathds{1}_{\rm A}}{d_{\rm A}} \otimes \Tr_{m_1 \rm{A}_1}(\bm{\Omega}_k) \label{eq:con_sym_2}\\
            \Tr_{n}(\leftindex_{m}{\bra{b}}\bm{\Omega}_k \ket{b}_{m}) &= \Tr_{nm\rm{B}}(\bm{\Omega}_{k}) \otimes \frac{\mathds{1}_{\rm B}}{md_{\rm B}}  \textnormal{ for all } b.  \label{eq:con_sym_3} 
        \end{align}
        If for every $k \in \mathbb{N}$ there is an operator $\bm{\Omega}_k$ that satisfies the properties \eqref{eq:con_sym_1}-\eqref{eq:con_sym_3} above and is such that 
        \begin{equation}
             \label{eq:effect_from_sym_ext}
            M^{\lambda} = \frac{1}{d_{\rm A}^{k-1}} \, \leftindex_{n}{\bra{\lambda}} \leftindex_{m_{k},m'}{\bra{\bm{\Phi}^+}}  \Tr_{m_{1}\rm{A}_1 \dots m_{k-1}\rm{A}_{k-1}}(\bm{\Omega}_k) \ket{\bm{\Phi}^+}_{m_{k},m'} \ket{\lambda}_{n},
        \end{equation}
        then $\{M^{\lambda}\}_{\lambda = 1}^{n} \in \textnormal{conv}(1\textnormal{R}-\textnormal{LOCC}_{m})$. The idea is now that $\bm{\Omega}_k$ can be understood as the block-diagonal operator that has all the operators $R^{\bm{a}\lambda b}$ on its diagonal. Explicitly, based on the assumed existence of operators $R^{\bm{a}\lambda b}$, we can construct the constrained symmetric extension $\bm{\Omega}_k$ as 
        \begin{equation}
            \label{eq:con_sym_ext_ansatz}
            \bm{\Omega}_k = \sum_{a_1\dots a_k = 1}^{m} \sum_{\lambda = 1}^{n} \sum_{b=1}^{m} \ketbra{\bm{a}}_{m_{1}\dots m_{k}} \otimes \ketbra{\lambda}_{n} \otimes \ketbra{b}_{m'} \otimes R^{\bm{a} \lambda b} 
        \end{equation}
        which directly satisfies the extension property of Eq.~\eqref{eq:effect_from_sym_ext} due to the assumptions of Eq.~\eqref{eq:1W_cond_start_copy} and \eqref{eq:1W_cond_end_copy}. We now check that $\bm{\Omega}_k$ is a proper constrained symmetric extension by verifying the equations \eqref{eq:con_sym_1}-\eqref{eq:con_sym_3}.  The property of Eq.~\eqref{eq:con_sym_1} can be shown by a straightforward calculation via 
        \begin{align}
            &(U_{k,md_{\rm A}}^{\sigma} \otimes \mathds{1}_{nm\rm{B}})\, \bm{\Omega}_k \, (U_{k,md_{\rm A}}^{\sigma^{-1}}\otimes \mathds{1}_{nm\rm{B}})\\ &= \sum_{a_1\dots a_k = 1}^{m} \sum_{\lambda = 1}^{n} \sum_{b=1}^{m} \ketbra{\sigma(\bm{a})} \otimes \ketbra{\lambda} \otimes \ketbra{b} \otimes (U_{k,d_{\rm A}}^{\sigma} \otimes \mathds{1}_{\rm B}) R^{\bm{a} \lambda b} (U_{k,d_{\rm A}}^{\sigma^{-1}} \otimes \mathds{1}_{\rm B}) \\
            &\underbrace{=}_{\textnormal{Eq.}~\eqref{eq:array_symmetry_copy}} \sum_{a_1\dots a_k = 1}^{m} \sum_{\lambda = 1}^{n} \sum_{b=1}^{m} \ketbra{\sigma(\bm{a})} \otimes \ketbra{\lambda} \otimes \ketbra{b} \otimes R^{\sigma(\bm{a}) \lambda b} = \bm{\Omega}_k.
        \end{align}
        Similarly, Eq.~\eqref{eq:con_sym_2} can be verified as follows:
        \begin{align}
            \Tr_{m_1}(\bm{\Omega}_k) &= \sum_{a_2\dots a_k = 1}^{m} \sum_{\lambda = 1}^{n} \sum_{b=1}^{m} \ketbra{a_2\dots a_k} \otimes \ketbra{\lambda} \otimes \ketbra{b} \otimes \sum_{a_{1}=1}^{m} R^{\bm{a} \lambda b} \\
            &\underbrace{=}_{\textnormal{Eq.}~\eqref{eq:left_array_cons_copy}} \sum_{a_2\dots a_k = 1}^{m} \sum_{\lambda = 1}^{n} \sum_{b=1}^{m} \ketbra{a_2\dots a_k} \otimes \ketbra{\lambda} \otimes \ketbra{b} \otimes \frac{\mathds{1}_{\rm A}}{d_{\rm A}} \otimes \sum_{a_1 = 1}^{m}\Tr_{\rm A_{1}}( R^{\bm{a} \lambda b}) \\
            &= \frac{\mathds{1}_{\rm A}}{d_{\rm A}} \otimes \Tr_{m_1 \rm{A}_1}(\bm{\Omega}_k).
        \end{align}
        Finally, also Eq.~\eqref{eq:con_sym_3} holds true as 
        \begin{align}
            \Tr_{n}(\leftindex_{m'}{\bra{b}}\bm{\Omega}_k \ket{b}_{m'}) &= \sum_{a_1\dots a_k = 1}^{m} \sum_{\lambda = 1}^{n} \Tr_{n}(\ketbra{\bm{a}}_{m_{1}\dots m_{k}} \otimes \ketbra{\lambda}_{n}  \otimes R^{\bm{a} \lambda b}) \\
            &= \sum_{a_1\dots a_k = 1}^{m} \sum_{\lambda = 1}^{n} \ketbra{\bm{a}}_{m_{1}\dots m_{k}} \otimes R^{\bm{a} \lambda b} \\
             &\underbrace{=}_{\textnormal{Eq.}~\eqref{eq:right_array_cons_copy}} \sum_{a_1\dots a_k = 1}^{m} \sum_{\lambda = 1}^{n} \sum_{b'=1}^{m} \ketbra{\bm{a}}_{m_{1}\dots m_{k}} \otimes \Tr_{\rm B}(R^{\bm{a} \lambda b'}) \otimes \frac{\mathds{1}_{\rm B}}{md_{\rm B}} \\
             &= \Tr_{nm\rm{B}}(\bm{\Omega}_{k}) \otimes \frac{\mathds{1}_{\rm B}}{md_{\rm B}}.
        \end{align}
        To conclude, we have shown that the existence of operators $R^{\bm{a} \lambda b}$ that satisfy the conditions \eqref{eq:1W_cond_start_copy}-\eqref{eq:1W_cond_end_copy} implies that the operator $\bm{\Omega}_k$ defined in Eq.~\eqref{eq:con_sym_ext_ansatz} is a proper constrained symmetric extension for each $k$ and hence we have $\{M^{\lambda}\}_{\lambda = 1}^{n} \in \textnormal{conv}(1\textnormal{R}-\textnormal{LOCC}_{m})$.
        \end{proof}

\section{Proof of Theorem \ref{thm:na_locc_m_fixed_char}}
\label{app:thm_2_proof}
Theorem \ref{thm:na_locc_m_fixed_char} stated the following.
\begin{theoremcopy}[$\textnormal{NA}-\textnormal{LOCC}_{m}$ hierarchy]
    A bipartite measurement $\{M^{\lambda}\}_{\lambda = 1}^{n} \subset \textnormal{Herm}(\mathbb{C}^{d_{\rm A}} \otimes \mathbb{C}^{d_{\rm B}})$ lies in the convex hull of $\textnormal{NA-LOCC}_{m}$ if and only if for each $k \in \mathbb{N}$ there exists an array of positive semidefinite operators $R^{\bm{a} \bm{b}} \in \textnormal{Herm}(\mathbb{C}^{d_{\rm A}^{k}} \otimes \mathbb{C}^{d_{\rm B}})$, for $\bm{a} = (a_1,\dots, a_k) \in \{1,\dots,m\}^{\times k}$ and $\bm{b} = (b_1, \dots, b_{m}) \in \{1,\dots,n\}^{\times m}$, that satisfies the following properties:
    \begin{align}
        &M^{\lambda} = \frac{1}{d_{\rm A}^{k-1}} \sum_{a_1, \dots, a_k = 1}^{m} \ \sum_{b_1,\dots,b_{m} = 1}^{n} \delta_{b_{a_{k}}, \lambda} \, \Tr_{\rm{A}_1 \dots \rm{A}_{k-1}}(R^{\bm{a} \bm{b}}) \textit{ for all } \lambda \label{eq:na_eff_red_copy}\\
        &\sum_{a_1=1}^{m} R^{a_1\dots a_k \bm{b}} = \frac{\mathds{1}_{\rm A}}{d_{\rm A}} \otimes \sum_{a_1 = 1}^{m}\Tr_{\rm A_{1}}( R^{a_1\dots a_k \bm{b}}) \textit{ for all } a_2,\dots, a_k, \bm{b} \\
        & \sum_{b_1, \dots, b_m = 1}^{n} R^{\bm{a}\bm{b}} = \sum_{b_1, \dots, b_m = 1}^{n} \Tr_{\rm B}(R^{\bm{a}\bm{b}}) \otimes \frac{\mathds{1}_{\rm B}}{d_{\rm B}} \textit{ for all } \bm{a} \\
        & (U_{k,d_{\rm A}}^{\sigma} \otimes \mathds{1}) \,  R^{\bm{a} \bm{b}} \, (U_{k,d_{\rm A}}^{{\sigma}^{-1}} \otimes \mathds{1}_{B}) = R^{\sigma(\bm{a}) \bm{b}} \textit{ for all } \sigma, \bm{a}, \bm{b} \\
        & {(R^{\bm{a}\bm{b}})}^{T_{\rm A^{\ell}}} \succcurlyeq 0 \textit{ for all } \bm{a}, \bm{b}, \ell \in \{0,\dots,k\} \\
        &\sum_{a_1 \dots a_k=1}^{m} \sum_{b_{1} \dots b_{m} = 1}^{n}\Tr(R^{\bm{a} \bm{b}}) = d_{\rm A}^k d_{\rm B}. \label{eq:na_norm_copy}
    \end{align}
\end{theoremcopy}
\begin{proof}[Proof of Theorem \ref{thm:na_locc_m_fixed_char}]
    The proof follows essentially the same logic as the proof of Theorem \ref{thm:1W_locc_m_fixed_char} presented in Appendix \ref{ap:theorem_1}. Let us consider a measurement $\{M^{\lambda}\}_{\lambda = 1}^{n} \in \text{conv}(\text{NA-LOCC}_{m})$. By definition, we can decompose the measurement operators via
    \begin{align}
            M^\lambda  &= \sum_{r} \sum_{a=1}^{m} \sum_{b_1,\dots,b_{m} = 1}^{n} q_{r}\delta_{b_a, \lambda} \, A_{r}^{a} \otimes   S_{r}^{\bm{b}} \\
            &= \sum_{r}  \sum_{a=1}^{m} q_{r} \leftindex_{m}{\bra{a}} \leftindex_{n_{a}}{\bra{\lambda}}  \Tr_{n_{1} \dots \cancel{n_{a}}\dots n_{m}}(\bm{A}^r \otimes \bm{B}^r) \ket{\lambda}_{n_{a}} \ket{a}_{m} 
        \end{align}
        where the operators $\bm{A}^r \in \textnormal{Herm}(\mathbb{C}^{m} \otimes \mathbb{C}^{d_{\rm A}})$ and $\bm{B}^r \in \textnormal{Herm}({(\mathbb{C}^{n})}^{\otimes m} \otimes \mathbb{C}^{d_{\rm B}})$ are given by 
        \begin{align}
            \bm{A}^r &= \sum_{a=1}^{m} \ketbra{a}_m \otimes A_{r}^{a}\\
            \bm{B}^r &= \sum_{b_{1}\dots b_{m} = 1}^{n} \ketbra{b_1}_{n_1} \otimes \dots \otimes \ketbra{b_{m}}_{n_{m}} \otimes S_{r}^{\bm{b}}.
        \end{align}
        These are positive semidefinite operators that are additionally subject to the affine-linear constraints 
        \begin{align}
            \Tr_{m}(\bm{A}^r) &= \sum_{a = 1}^{m} A_{r}^a = \mathds{1}_{\rm A} \\
             \Tr_{n_1 \dots n_{m}}(\bm{B}^r) &= \sum_{b_{1}\dots b_{m} = 1}^{n} S_{r}^{\bm{b}} = \mathds{1}_{\rm B}.
        \end{align}
    Hence, a measurement is in the convex hull of non-adaptive LOCC measurements if and only if there is a constrained separable operator $\bm{X} \in \textnormal{Herm}(\mathbb{C}^{md_{\rm A}} \otimes \mathbb{C}^{n^m d_{\rm B}})$ of the form 
        \begin{equation}
           \bm{X} = \sum_{r} q_{r} \, \bm{A}^{r} \otimes \bm{B}^{r}
        \end{equation}
        such that 
        \begin{equation}
            M^{\lambda} =   \sum_{a=1}^{m}  \leftindex_{m_{k}}{\bra{a}} \leftindex_{n_{a}}{\bra{\lambda}}  \Tr_{\rm n_{1} \dots \cancel{n_{a}}\dots n_{m}}(\bm{X}) \ket{\lambda}_{n_{a}} \ket{a}_{m_{k}}
        \end{equation}
        for all $\lambda$. 
        Equivalently, by using the framework presented in Appendix \ref{app:sym_ext}, these measurements can be identified with constrained $k$-symmetric extensions $\bm{\Omega}_{k} \in \textnormal{Herm}({(\mathbb{C}^{md_{\rm A}})}^{\otimes k} \otimes \mathbb{C}^{n^m d_{\rm B}})$ that satisfy 
        \begin{align}
            \bm{\Omega}_k  &= (U_{k,md_{\rm A}}^{\sigma} \otimes \mathds{1}_{n_1 \dots n_m \rm{B}})\, \bm{\Omega}_k \, (U_{k,md_{\rm A}}^{\sigma^{-1}}\otimes \mathds{1}_{n_1 \dots n_m \rm{B}}) \textnormal{ for all } \sigma \in S_k \label{eq:na_con_sym_1}\\
            \Tr_{m_1} (\bm{\Omega}_k) &= \frac{\mathds{1}_{\rm A}}{d_{\rm A}} \otimes \Tr_{m_1 \rm{A}_1}(\bm{\Omega}_k) \label{eq:na_con_sym_2}\\
            \Tr_{n_1 \dots n_{m}}(\bm{\Omega}_k) &= \Tr_{n_1 \dots n_{m} \rm{B}}(\bm{\Omega}_{k}) \otimes \frac{\mathds{1}_{\rm B}}{d_{\rm B}}  \label{eq:na_con_sym_3} 
        \end{align}
        and are chosen to produce the correct measurement operators via
        \begin{equation}
            \label{eq:na_extension_crit}
            M^{\lambda} =   \sum_{a=1}^{m}  \leftindex_{m_{k}}{\bra{a}} \leftindex_{n_{a}}{\bra{\lambda}}  \Tr_{\rm m_{1}A_{1} \dots m_{k-1}A_{k-1}n_{1} \dots \cancel{n_{a}}\dots n_{m}}(\bm{\Omega}_{k}) \ket{\lambda}_{n_{a}} \ket{a}_{m_{k}}. 
        \end{equation}
        It is now straightforward to show that the ansatz for $\bm{\Omega}_{k}$ given by 
        \begin{equation}
            \bm{\Omega}_{k} = \sum_{a_1 \dots a_k = 1}^{m} \sum_{b_{1} \dots b_{m} = 1}^{n} \ketbra{\bm{a}}_{m_1\dots m_{k}} \otimes \ketbra{\bm{b}}_{n_{1}\dots n_{m}} \otimes R^{\bm{a} \bm{b}}
        \end{equation}
        satisfies the constraints \eqref{eq:na_con_sym_1}-\eqref{eq:na_extension_crit}, given that the positive semidefinite operators $R^{\bm{a} \bm{b}}$ satisfy the constraints \eqref{eq:na_eff_red_copy}-\eqref{eq:na_norm_copy}.

\end{proof}

\printbibliography
\end{document}